\documentclass[american,11pt]{article}
\usepackage{}
\usepackage[runin]{abstract}
\usepackage{amsfonts}
\usepackage{amsmath}
\usepackage{amssymb}
\usepackage{amsthm}
\usepackage{appendix}
\usepackage{babel}
\usepackage[a4paper, top=1.2cm, bottom=1.2cm, left=2.7cm, right=2.7cm,
            includefoot, includehead]{geometry}
\usepackage[bookmarks=true, bookmarksnumbered=false,
 bookmarksopen=true, breaklinks=false, pdfborder={0 0 1},
 backref=false, colorlinks=false]{hyperref}
\hypersetup{pdftitle={The Multiple Permutation Problem and Some Conjectures},
            pdfauthor={Zan Pan},
            pdfsubject={enumerative combinatorics},
            pdfkeywords={multiple permutation, enumeration, asymptotic formula}}
\usepackage[center, small, bf, indentafter]{titlesec}
\titlelabel{\thetitle.\quad}

\abslabeldelim{.\quad}
\newcommand{\ang}[1]{\ensuremath{\langle #1 \rangle}}

\theoremstyle{plain}
\newtheorem{theorem}{Theorem}
\newtheorem{conjecture}[theorem]{Conjecture}

\begin{document}
\title{\normalsize{\textbf{THE MULTIPLE PERMUTATION PROBLEM AND \\
       SOME CONJECTURES}}}
\author{\small{Zan Pan}\\
 \href{mailto:panzan@mail.ustc.edu.cn}{\small{panzan@mail.ustc.edu.cn}}\\
 \textsc{\footnotesize{University of Science and Technology of China,
  Hefei, 230026, China}}}
\date{}
\maketitle

\begin{abstract}
 In this paper, we proposed an interesting problem that might be classified into
 enumerative combinatorics. Featuring a distinctive two-fold dependence upon the
 sequences' terms, our problem can be really difficult, which calls for novel
 approaches to work it out for any given pair $(m,n)$. Complete or partial
 solutions for $m=2$, 3 with smaller $n$'s are listed. Moreover, we have proved
 the necessary condition for $p(m,n)\neq 0$ and suggested an elegant asymptotic
 formula for $p(2,n)$. In addition, several challenging conjectures are provided,
 together with concise comments.
\end{abstract}

\section{Introduction}

 Permutation, as one of the classical topics in combinatorics, has been studied
 in different contexts and brought about numerous applications \cite{Bona, Knu98}.
 By contrast, the investigations on its generalization to multisets seem
 insufficient. Perhaps this is mainly due to the shortage of unique methods
 and fascinating problems. Primarily motivated by an illustrative example
 of the transfer matrix method in R.~P.~Stanley's notable book \cite{Stan97},
 we proposed its variant generalization to multisets, which is stated as
 follows:
\begin{quote}
  Let $[n]$ represent $\{1,2,\ldots,n\}$ and $\mathfrak{S}_n$ denote the set of
  all its  permutations. Given a multiset $M=[n]^m$ with the multiplicity $m$,
  find  the set of all its permutations denoted by $\mathfrak{S}_n^m$, in which
  any two identical elements are separated by as many terms as their own value.
  For instance, 3 1 2 1 3 2 is the simplest case.  To put it differently,
  we use $\varpi(i)$ to represent the $i$-th element of a sequence. Then
  $\mathfrak{S}_n^m$ is isomorphic to $\{\mathcal{S}:\mathcal{S}\in
  \mathfrak{S}_{mn}\}$, for whose elements the following conditions hold
 \[
  \begin{cases}
   \varpi(i)+n=\varpi\big(i+\varpi(i)+1\big),
    & \text{if } i+\varpi(i)+1\leqslant mn\\
   \varpi(i)-n=\varpi\big(n+i-\varpi(i)-1\big),
    & \text{otherwise}
  \end{cases}
 \]

 For the convenience of discussion, we would like to exert another restriction:
 if two solutions,  such as 4 1 3 1 2 4 3 2 and 2 3 4 2 1 3 1 4, can be mutually
 generated by an inversion of each other, only the one with smaller index
 of the first $n$ or the same index but greater initial element is counted.
 The notation $p(m,n)$ will be consistently used to represent the result
 of such an enumeration.
\end{quote}

 Although it may be somewhat confusing, we still call it the multiple permutation
 problem.  Challenge for us is to find out all the solutions for any given pair
 $(m,n)$ and to explore the mathematical structure hidden behind those sequences.
 Special attention should be paid to the observation that the conditions described
 above have a two-fold dependence on the sequences' terms, which makes the problem
 really difficult to solve. It seems that there are no effective methods available
 at present. We believe that a deep survey into this problem will yield a dozen
 of nontrivial results and help us develop ingenious ideas. Just like D.~Foata's
 introduction of the intercalation product which extends many known results
 about ordinary permutations to multiset permutations~\cite{Knu98}, novel
 approaches to this problem are certainly conducive to figuring out
 many enumerative problems regarding multisets.

\section{\label{sec:search}Search for Solutions}

 It is easy to see that only one solution exists for $(2,3)$ and $(2,4)$,
 but things become complicate for a larger $n$ or other $m$'s.
 With the help of computer programming using trace-back algorithm,
 we have obtained complete or partial solutions for certain pairs $(m,n)$.
 For $m>3$, the required computing time is too long, so we have not taken into
 account those cases. The following offers a glance of selected results:
\begin{quote}
  7 3 6 2 5 3 2 4 7 6 5 1 4 1\\
  5 7 4 1 6 1 5 4 3 7 2 6 3 2\\
  2 7 4 2 3 5 6 4 3 7 1 5 1 6\\
  2 4 7 2 3 6 4 5 3 1 7 1 6 5\\
  8 3 7 2 6 3 2 4 5 8 7 6 4 1 5 1\\
  2 8 3 2 4 6 3 7 5 4 8 1 6 1 5 7\\
  7 2 8 3 2 4 6 3 7 5 4 8 1 6 1 5\\
  4 2 6 8 2 4 7 5 1 6 1 3 8 5 7 3\\
  11 6 10 2 9 3 2 8 6 3 7 5 11 10 9 4 8 5 7 1 4 1\\
  1 3 1 4 11 3 6 7 4 9 10 8 5 6 2 7 11 2 5 9 8 10\\
  12 10 11 6 4 5 8 9 7 4 6 5 10 12 11 8 7 9 3 1 2 1 3 2\\
  6 1 3 1 8 12 3 6 7 10 11 9 2 8 4 2 7 5 12 4 10 9 11 5\\
  15 13 14 8 5 12 7 11 4 10 5 9 8 4 7 13 15 14 12 11 10 9 6 3 1 2 1 3 2 6\\
  4 1 5 1 2 4 15 2 5 7 10 8 11 13 14 12 9 7 6 3 8 10 15 3 11 6 9 13 12 14\\
  16 14 15 9 7 13 3 12 6 11 3 10 7 9 8 6 14 16 15 13 12 11 10 8 5 2 4 1 2 1 5 4\\
  9 3 1 4 1 3 10 16 4 6 9 11 12 14 15 13 6 10 2 7 5 2 8 11 16 12 5 7 14 13 15 8\\
  19 17 18 14 8 16 9 15 6 1 13 1 12 8 11 6 9 10 14 17 19 18 16 15 13 12 11 7 10 3
  5 2 4 3 2 7 5 4\\
  20 18 19 15 11 17 10 16 9 5 14 1 13 1 12 5 11 10 9 15 18 20 19 17 16 14 13 12 8
  4 7 3 6 2 4 3 2 8 7 6\\
  21 14 11 19 10 12 13 17 15 16 20 22 23 18 11 10 14 9 12 8 13 7 21 19 15 17 16 9
  8 7 6 20 18 5 22 4 23 6 3 5 4 2 3 1 2 1\\
  22 16 12 17 19 11 13 14 18 21 23 24 20 15 10 12 9 11 16 8 13 17 14 22 19 10 9 18
  8 15 7 21 6 20 23 5 24 4 7 6 3 5 4 2 3 1 2 1\\
  \\
  1 9 1 6 1 8 2 5 7 2 6 9 2 5 8 4 7 6 3 5 4 9 3 8 7 4 3\\
  1 9 1 2 1 8 2 4 6 2 7 9 4 5 8 6 3 4 7 5 3 9 6 8 3 5 7\\
  3 4 7 9 3 6 4 8 3 5 7 4 6 9 2 5 8 2 7 6 2 5 1 9 1 8 1\\
  1 10 1 6 1 7 9 3 5 8 6 3 10 7 5 3 9 6 8 4 5 7 2 10 4 2 9 8 2 4\\
  1 10 1 2 1 4 2 9 7 2 4 8 10 5 6 4 7 9 3 5 8 6 3 10 7 5 3 9 6 8\\
  4 10 1 7 1 4 1 8 9 3 4 7 10 3 5 6 8 3 9 7 5 2 6 10 2 8 5 2 9 6\\
  8 1 10 1 3 1 9 6 3 8 4 7 3 10 6 4 9 5 8 7 4 6 2 5 10 2 9 7 2 5\\
  1 3 1 10 1 3 4 9 6 3 8 4 5 7 10 6 4 9 5 8 2 7 6 2 5 10 2 9 8 7\\
  17 15 3 16 9 10 3 1 12 1 3 1 13 14 9 6 10 15 17 5 16 12 6 11 9 5 13 10 14 6 7 5
  8 15 12 11 17 16 7 4 13 8 2 14 4 2 7 11 2 4 8\\
  9 3 1 4 1 3 1 17 4 3 9 14 10 4 15 16 11 6 8 13 9 7 12 10 6 17 14 8 11 7 15 6 16
  13 10 12 8 7 5 2 11 14 2 17 5 2 15 13 12 16 5\\
  18 16 5 17 11 4 2 9 5 2 4 14 2 15 5 4 11 9 16 18 12 17 13 6 7 8 14 9 11 15 6 10
  7 12 8 16 13 6 18 17 7 14 10 8 3 15 12 1 3 1 13 1 3 10\\
  4 12 10 7 13 4 2 18 14 2 4 7 2 10 12 6 16 17 13 7 15 9 6 14 10 11 18 12 5 6 8 9
  13 16 5 17 15 11 14 8 5 9 1 3 1 18 1 3 8 11 16 3 15 17\\
  19 17 13 18 4 11 8 2 16 4 2 9 15 2 4 8 13 11 14 17 19 9 18 12 8 16 5 7 15 11 13
  9 5 14 10 7 12 17 5 6 19 18 16 7 15 10 6 3 14 12 1 3 1 6 1 3 10\\
  10 1 3 1 13 1 3 7 19 14 3 10 8 15 12 7 17 18 13 5 16 8 10 7 14 5 11 12 19 15 8 5
  13 9 17 6 18 16 11 14 12 2 6 9 2 15 4 2 19 6 11 4 17 9 16 18 4
\end{quote}

  A much more complete list of these solutions can be found in the appendix.
  Now we only focus on the analysis of its enumeration.
  It comes as a surprise that $p(m,n)$ increases extremely fast with
  respect to $n$:

\begin{center}
\begin{tabular}{lll}
  &$p(2,7)=26$                           &$p(3,8)=0$\\
  &$p(2,8)=150$                          &$p(3,9)=3$\\
  &$p(2,11)=17792$                       &$p(3,10)=5$\\
  &$p(2,12)=108144$                      &$p(3,17)=13440$\\
  &$p(2,15)=39809640$                    &$p(3,18)=54947$\\
  &$p(2,16)=326721800$\qquad\qquad\qquad &$p(3,19)=249280$
\end{tabular}
\end{center}

\section{Existence and Asymptosy}

 From the data above, we could see clearly not consecutive $n$'s are included,
 which means that only for a certain type of integers can the problem exist
 solutions. A brief analysis conforms this point.

\begin{theorem}
 \label{prime}
 If $m$ is a prime, then the necessary condition for $p(m,n)\neq0$
 is that $n\equiv-1$, $0$, $1$, $\dots$, $(m-2)(\!\!\!\mod m^2)$.
\end{theorem}
\begin{proof}
 For a given solution $\mathcal{S}$, let $\theta_k$ denote the index of integer 
 $k$'s first occurrence in it. Then $\theta_k+k+1$, $\theta_k+2(k+1)$, $\dots$,
 $\theta_k+(m-1)(k+1)$ are indices of other successive $k$'s. Partition the
 sequence into $m$ subsequences according to the remainders of each
 element's index divided by $m$. For example, 1 9 1 2 1 8 2 4 6 2 7 9 4
 5 8 6 3 4 7 5 3 9 6 8 3 5 7 can be partitioned into 1 2 2 2 4 6 7 9
 3 $\sqcup$ 9 1 4 7 5 3 5 6 5 $\sqcup$ 1 8 6 9 8 4 3 8 7, where $\sqcup$
 is defined as an operator that indicates the equivalence of there subsequences
 to the solution $\mathcal{S}$. We may call it the congruential representation.
 If $k\neq m-1$, then $[\theta_k]$, $[\theta_k+k+1]$, $\dots$,
 $[\theta_k+(m-1)(k+1)]$ forms a complete residue system with the modulus $m$,
 so each subsequence can get a $k$; if $k=m-1$, then $[\theta_k]=[\theta_k+k+1]
 =\cdots=[\theta_k+(m-1)(k+1)]$, and only one subsequence possesses all the $k$'s.
 Therefore, there must exist $mt$ $(t\in\!\mathbb{N})$ integers whose remainders
 with the modulus $m$ are $-1$ so that all subsequences have the same many
 elements. This requires the claimed distribution for the values of $n$.
\end{proof}
 It is enlightening to emphasize that $m$ must be a prime in Theorem~\ref{prime}.
 If not, the condition for solutions' existence might be totally different.
 Set $m=4$, and through a similar but more complex discussion we can come to
 the following proposition.
\newtheorem{proposition}[theorem]{Proposition}
\begin{proposition}
  The necessary condition for $p(4,n)\neq0$ is $n=8t$ or $n=8t-1$
  $(t\in\!\mathbb{N},t\geqslant2)$.
\end{proposition}
\begin{proof}
 For brevity, we use $\langle0\rangle$, $\langle1\rangle$, $\langle2\rangle$,
 $\langle3 \rangle$ to represent the four subsequences respectively and $[0]$,
 $[1]$, $[2]$, $[3]$ to denote numbers whose residues divided by 4 are 0, 1,
 2, 3, correspondingly. Suppose that, in the consecutive numbers from 1 to
 $n$, there are $\lambda_0\,[0]$'s, $\lambda_1\,[1]$'s, $\lambda_2\,[2]$'s and
 $\lambda_3\,[3]$'s. It's easy to understand that $[0]$'s and $[2]$'s are
 equally distributed in four subsequences and $[1]$'s equally distributed
 only in two while $[3]$'s just in one. Therefore, we have
\begin{eqnarray*}
  \ang{0}=&\hspace{-7pt}\big\{\lambda_0\cdot[0],\,2x\cdot[1],\,
   \lambda_2\cdot[2],\,4u\cdot[3]\big\},\\
  \ang{1}=&\hspace{-7pt}\big\{\lambda_0\cdot[0],\,2y\cdot[1],\,
   \lambda_2\cdot[2],\,4v\cdot[3]\big\},\\
  \ang{2}=&\hspace{-7pt}\big\{\lambda_0\cdot[0],\,2x\cdot[1],\,
   \lambda_2\cdot[2],\,4u\cdot[3]\big\},\\
  \ang{3}=&\hspace{-7pt}\big\{\lambda_0\cdot[0],\,2y\cdot[1],\,
   \lambda_2\cdot[2],\,4v\cdot[3]\big\},
\end{eqnarray*}
 where $x$, $y$, $u$, $v$ are all positive integers. Then it follows that
 $x+2u=y+2v$, $\lambda_1=x+y$, $\lambda_3=2(u+v)$. Obviously, $\lambda_1$,
 $\lambda_3$ are both even numbers and $\lambda_3\geqslant4$. Now we have
 arrived at the conclusion.
\end{proof}
\begin{conjecture}
 The condition in Theorem~\ref{prime} is also sufficient for large enough $n$'s.
\end{conjecture}
 The list appearing in the end of last section lends strong credence to this
 assertion, despite that it may be hard to offer a proof. As to a composite
 number assigned to $m$, no such trends exist at hand. On the contrary, it is
 quite likely that there are no solutions at all. Another challenge lies in
 the asymptotic expression for $p(2,n)$.
\begin{conjecture}
 $p(2,n)\sim n!/2^n$.
\end{conjecture}
 Considering $n=12$, the asymptotic formula gives 116944.75, reasonably close
 to the exact result. Similarly, for $m=3$, we have $p(3,n)\sim n!/2^{2n+1}$.
 However, it does not gain sufficient support. This maybe reminds us of G.~H.~Hardy
 and S.~Ramanujan's famous asymptotic formula on partition function \cite{HR18},
 albeit the elegance and simplicity of the expression, behind it hides an extremely
 complex convergent series that was obtained by H.~Rademacher \cite{Rad37}.

\begin{proposition}
 $p(m,n)<n!/2$.
\end{proposition}
\begin{proof}
  Suppose that $\phi_k$ denotes the term of the $k$-th distinctive number that first
 appears in the sequence, then we will call $\phi_1\,\phi_2\,\phi_3\cdots\phi_n$ the
 $\phi$\,-sequence of the solution $\mathcal{S}$. It's easy to understand that this kind of
 sequence is just a permutation of consecutive numbers from 1 to $n$. Consider the
 restrictions, and then we complete this proof.
\end{proof}

\section{Other Topics}

 In this section, we would like to give some related questions that are not mentioned
 above. Since the problem can widely involve combinatorics, number theory and computer
 algorithms, it cannot rule out the probability of omitting certain ones, even some
 profound insights.
\begin{itemize}
\item What is about the complexity of this problem and whether there exists an
  effective algorithm so that we can quickly get the number of total solutions?
  It should be pointed out that the method adopted in section~\ref{sec:search}
  wastes considerable computing time. For $m=2$, the complexity is $O[n(n+1)
  (n+2)\cdots(2n-2)\cdot n/2]=O[(4n)^ne^{-n}]$.
 \item Try to discover the generating function for $p(m,n)$ if any and reveal the
  sequence's properties like Conjecture~\ref{conj:mu}, or its potential links to
  the distribution of prime numbers.
 \item Does this problem have certain connections with other branches of mathematics,
  such as graph theory and the symmetric group?
\end{itemize}
\begin{conjecture}
 \label{conj:mu}
 Let $\mathcal{S}$ represent a solution to $(2,n)$ and $\mu_i$ denote the
 number of all integers $i$ satisfying $\theta_k\leqslant i\leqslant\theta_k+k+1$.
 We will call $\mu_1\,\mu_2\cdots\mu_{2n}$ the $\mu$-sequence of $\mathcal{S}$
 and it is easy to deduce that $\sum_{i=1}^{2n} \mu_i =\frac{1}{2}(n^2+5n)$.
 Show that there must exist $i \leqslant n$ such that $\mu_i=\lfloor
 \frac{1}{2}(n+3)\rfloor$, and $\mu_1<\mu_2\leqslant\cdots\leqslant
 \mu_{i-1}\leqslant\mu_i\geqslant\mu_{i+1}\geqslant\cdots\geqslant\mu_{2n-1}
 >\mu_{2n}$. In other words, $\mu$-sequence is an unimodal sequence with the weight
 $\frac{1}{2}(n^2+5n)$.
\end{conjecture}

\addcontentsline{toc}{section}{References}
\phantomsection

\newpage
\appendix
\addappheadtotoc
\section*{Appendix}

\renewcommand{\qquad}{\hspace{32pt}}
\noindent
$p(2,7)=26$. The following is a complete list of all its solutions:
\medskip \\
7 3 6 2 5 3 2 4 7 6 5 1 4 1 \qquad 7 2 6 3 2 4 5 3 7 6 4 1 5 1 \qquad
7 2 4 6 2 3 5 4 7 3 6 1 5 1 \\
7 3 1 6 1 3 4 5 7 2 6 4 2 5 \qquad 7 1 4 1 6 3 5 4 7 3 2 6 5 2 \qquad
7 1 3 1 6 4 3 5 7 2 4 6 2 5 \\
7 4 1 5 1 6 4 3 7 5 2 3 6 2 \qquad 7 2 4 5 2 6 3 4 7 5 3 1 6 1 \qquad
5 7 2 6 3 2 5 4 3 7 6 1 4 1 \\
3 7 4 6 3 2 5 4 2 7 6 1 5 1 \qquad 5 7 4 1 6 1 5 4 3 7 2 6 3 2 \qquad
5 7 2 3 6 2 5 3 4 7 1 6 1 4 \\
1 7 1 2 6 4 2 5 3 7 4 6 3 5 \qquad 5 7 1 4 1 6 5 3 4 7 2 3 6 2 \qquad
1 7 1 2 5 6 2 3 4 7 5 3 6 4 \\
2 7 4 2 3 5 6 4 3 7 1 5 1 6 \qquad 6 2 7 4 2 3 5 6 4 3 7 1 5 1 \qquad
2 6 7 2 1 5 1 4 6 3 7 5 4 3 \\
3 6 7 1 3 1 4 5 6 2 7 4 2 5 \qquad 5 1 7 1 6 2 5 4 2 3 7 6 4 3 \qquad
2 3 7 2 6 3 5 1 4 1 7 6 5 4 \\
4 1 7 1 6 4 2 5 3 2 7 6 3 5 \qquad 5 2 7 3 2 6 5 3 4 1 7 1 6 4 \qquad
3 5 7 4 3 6 2 5 4 2 7 1 6 1 \\
3 5 7 2 3 6 2 5 4 1 7 1 6 4 \qquad 2 4 7 2 3 6 4 5 3 1 7 1 6 5 \\

\noindent
$p(2,8)=150$. The following is a complete list of all its solutions:
\medskip \\
 8 3 7 2 6 3 2 4 5 8 7 6 4 1 5 1 \quad 8 2 7 3 2 6 4 3 5 8 7 4 6 1 5 1 \quad
 8 2 7 1 2 1 6 4 5 8 7 3 4 6 5 3 \\
 8 2 7 1 2 1 5 6 4 8 7 3 5 4 6 3 \quad 8 4 5 7 2 6 4 2 5 8 3 7 6 1 3 1 \quad
 8 2 4 7 2 6 3 4 5 8 3 7 6 1 5 1 \\
 8 5 2 7 3 2 6 5 3 8 4 7 1 6 1 4 \quad 8 3 5 7 2 3 6 2 5 8 4 7 1 6 1 4 \quad
 8 4 2 7 5 2 4 6 3 8 5 7 3 1 6 1 \\
 8 2 3 7 2 4 3 5 6 8 4 7 1 5 1 6 \quad 8 6 4 2 7 5 2 4 6 8 3 5 7 1 3 1 \quad
 8 6 3 1 7 1 3 5 6 8 4 2 7 5 2 4 \\
 8 4 2 6 7 2 4 3 5 8 6 3 7 1 5 1 \quad 8 1 2 1 7 2 6 3 5 8 4 3 7 6 5 4 \quad
 8 4 5 1 7 1 4 6 5 8 2 3 7 2 6 3 \\
 8 3 5 2 7 3 2 6 5 8 4 1 7 1 6 4 \quad 8 1 3 1 7 4 3 5 6 8 4 2 7 5 2 6 \quad
 8 1 2 1 7 2 4 5 6 8 3 4 7 5 3 6 \\
 8 6 4 2 5 7 2 4 6 8 5 3 1 7 1 3 \quad 8 6 1 3 1 7 5 3 6 8 4 2 5 7 2 4 \quad
 8 4 5 6 2 7 4 2 5 8 6 3 1 7 1 3 \\
 8 4 1 6 1 7 4 3 5 8 6 3 2 7 5 2 \quad 8 1 4 1 6 7 3 4 5 8 3 6 2 7 5 2 \quad
 8 1 3 1 5 7 3 4 6 8 5 2 4 7 2 6 \\
 8 2 4 6 2 5 7 4 3 8 6 5 3 1 7 1 \quad 8 1 2 1 6 2 7 5 3 8 4 6 3 5 7 4 \quad
 8 5 1 4 1 6 7 5 4 8 2 3 6 2 7 3 \\
 8 2 5 3 2 6 7 3 5 8 4 1 6 1 7 4 \quad 8 2 3 6 2 5 3 7 4 8 6 5 1 4 1 7 \quad
 8 3 1 6 1 3 5 7 4 8 6 2 5 4 2 7 \\[5pt]
 8 1 2 1 6 2 5 7 4 8 3 6 5 4 3 7 \quad 8 1 3 1 5 6 3 7 4 8 5 2 6 4 2 7 \quad
 7 8 2 3 6 2 5 3 7 4 8 6 5 1 4 1 \\
 7 8 3 1 6 1 3 5 7 4 8 6 2 5 4 2 \quad 7 8 1 2 1 6 2 5 7 4 8 3 6 5 4 3 \quad
 7 8 1 3 1 5 6 3 7 4 8 5 2 6 4 2 \\
 6 8 2 7 3 2 5 6 3 4 8 7 5 1 4 1 \quad 3 8 4 7 3 6 2 4 5 2 8 7 6 1 5 1 \quad
 3 8 5 7 3 1 6 1 5 4 8 7 2 6 4 2 \\
 3 8 4 7 3 2 6 4 2 5 8 7 1 6 1 5 \quad 5 8 2 7 4 2 5 6 3 4 8 7 3 1 6 1 \quad
 5 8 1 7 1 3 5 6 4 3 8 7 2 4 6 2 \\
 4 8 5 7 1 4 1 6 5 3 8 7 2 3 6 2 \quad 6 8 3 1 7 1 3 6 4 5 8 2 7 4 2 5 \quad
 4 8 6 2 7 4 2 3 5 6 8 3 7 1 5 1 \\
 2 8 5 2 7 1 6 1 5 4 8 3 7 6 4 3 \quad 2 8 5 2 7 3 4 6 5 3 8 4 7 1 6 1 \quad
 4 8 3 5 7 4 3 6 2 5 8 2 7 1 6 1 \\
 5 8 4 1 7 1 5 4 6 3 8 2 7 3 2 6 \quad 5 8 2 3 7 2 5 3 6 4 8 1 7 1 4 6 \quad
 6 8 5 2 4 7 2 6 5 4 8 3 1 7 1 3 \\
 5 8 6 4 2 7 5 2 4 6 8 3 1 7 1 3 \quad 2 8 6 2 1 7 1 4 5 6 8 3 4 7 5 3 \quad
 5 8 1 4 1 7 5 6 4 2 8 3 2 7 6 3 \\
 1 8 1 5 3 7 4 6 3 5 8 4 2 7 6 2 \quad 2 8 1 2 1 7 4 6 3 5 8 4 3 7 6 5 \quad
 1 8 1 3 4 7 5 3 6 4 8 2 5 7 2 6 \\
 2 8 1 2 1 7 5 3 6 4 8 3 5 7 4 6 \quad 6 8 1 4 1 5 7 6 4 3 8 5 2 3 7 2 \quad
 2 8 6 2 3 5 7 4 3 6 8 5 4 1 7 1 \\[5pt]
 4 8 5 3 6 4 7 3 5 2 8 6 2 1 7 1 \quad 2 8 5 2 6 3 7 4 5 3 8 6 4 1 7 1 \quad
 1 8 1 4 6 3 7 5 4 3 8 6 2 5 7 2 \\
 2 8 5 2 4 6 7 3 5 4 8 3 6 1 7 1 \quad 3 8 4 5 3 6 7 4 2 5 8 2 6 1 7 1 \quad
 1 8 1 5 2 6 7 2 4 5 8 3 6 4 7 3 \\
 2 8 4 2 3 6 7 4 3 5 8 1 6 1 7 5 \quad 2 8 1 2 1 6 7 3 4 5 8 3 6 4 7 5 \quad
 3 8 2 5 3 2 7 4 6 5 8 1 4 1 7 6 \\
 2 8 1 2 1 5 7 4 6 3 8 5 4 3 7 6 \quad 3 8 6 2 3 5 2 7 4 6 8 5 1 4 1 7 \quad
 3 8 6 1 3 1 5 7 4 6 8 2 5 4 2 7 \\
 5 8 2 4 6 2 5 7 4 3 8 6 1 3 1 7 \quad 4 8 5 2 6 4 2 7 5 3 8 6 1 3 1 7 \quad
 5 8 1 4 1 6 5 7 4 3 8 2 6 3 2 7 \\
 3 8 5 2 3 6 2 7 5 4 8 1 6 1 4 7 \quad 2 8 3 2 4 6 3 7 5 4 8 1 6 1 5 7 \quad
 2 8 1 2 1 6 4 7 5 3 8 4 6 3 5 7 \\
 4 8 1 5 1 4 6 7 3 5 8 2 3 6 2 7 \quad 2 8 1 2 1 5 6 7 3 4 8 5 3 6 4 7 \quad
 7 3 8 6 2 3 5 2 7 4 6 8 5 1 4 1 \\
 7 3 8 6 1 3 1 5 7 4 6 8 2 5 4 2 \quad 7 5 8 2 4 6 2 5 7 4 3 8 6 1 3 1 \quad
 7 4 8 5 2 6 4 2 7 5 3 8 6 1 3 1 \\
 7 5 8 1 4 1 6 5 7 4 3 8 2 6 3 2 \quad 7 3 8 5 2 3 6 2 7 5 4 8 1 6 1 4 \quad
 7 2 8 3 2 4 6 3 7 5 4 8 1 6 1 5 \\
 7 2 8 1 2 1 6 4 7 5 3 8 4 6 3 5 \quad 7 4 8 1 5 1 4 6 7 3 5 8 2 3 6 2 \quad
 7 2 8 1 2 1 5 6 7 3 4 8 5 3 6 4 \\[5pt]
 5 7 8 4 2 6 5 2 4 7 3 8 6 1 3 1 \quad 4 7 8 2 5 4 2 6 3 7 5 8 3 1 6 1 \quad
 2 7 8 2 3 4 5 6 3 7 4 8 5 1 6 1 \\
 3 7 8 2 3 4 2 5 6 7 4 8 1 5 1 6 \quad 6 4 8 5 7 2 4 6 2 5 3 8 7 1 3 1 \quad
 2 3 8 2 7 3 6 1 5 1 4 8 7 6 5 4 \\
 5 1 8 1 7 2 5 6 2 3 4 8 7 3 6 4 \quad 4 1 8 1 7 4 2 5 6 2 3 8 7 5 3 6 \quad
 6 1 8 1 4 7 3 6 5 4 3 8 2 7 5 2 \\
 6 2 8 4 2 7 3 6 4 5 3 8 1 7 1 5 \quad 2 6 8 2 1 7 1 4 6 5 3 8 4 7 3 5 \quad
 3 5 8 6 3 7 1 5 1 4 6 8 2 7 4 2 \\
 3 4 8 6 3 7 4 1 5 1 6 8 2 7 5 2 \quad 5 1 8 1 3 7 5 6 3 2 4 8 2 7 6 4 \quad
 3 4 8 5 3 7 4 6 1 5 1 8 2 7 6 2 \\
 2 4 8 2 3 7 4 6 3 5 1 8 1 7 6 5 \quad 5 2 8 3 2 7 5 3 6 4 1 8 1 7 4 6 \quad
 2 5 8 2 4 7 3 5 6 4 3 8 1 7 1 6 \\
 3 5 8 2 3 7 2 5 6 4 1 8 1 7 4 6 \quad 3 4 8 5 3 7 4 2 6 5 2 8 1 7 1 6 \quad
 3 1 8 1 3 7 5 2 6 4 2 8 5 7 4 6 \\
 6 2 8 5 2 4 7 6 3 5 4 8 3 1 7 1 \quad 6 3 8 4 5 3 7 6 4 2 5 8 2 1 7 1 \quad
 6 1 8 1 5 3 7 6 4 3 5 8 2 4 7 2 \\
 4 6 8 3 5 4 7 3 6 2 5 8 2 1 7 1 \quad 2 6 8 2 5 3 7 4 6 3 5 8 4 1 7 1 \quad
 3 6 8 1 3 1 7 5 6 2 4 8 2 5 7 4 \\
 3 6 8 1 3 1 7 4 6 5 2 8 4 2 7 5 \quad 4 5 8 6 3 4 7 5 3 2 6 8 2 1 7 1 \quad
 3 1 8 1 3 6 7 2 4 5 2 8 6 4 7 5 \\[5pt]
 2 3 8 2 4 3 7 5 6 4 1 8 1 5 7 6 \quad 3 1 8 1 3 4 7 5 6 2 4 8 2 5 7 6 \quad
 4 6 8 2 5 4 2 7 6 3 5 8 1 3 1 7 \\
 3 6 8 1 3 1 5 7 6 4 2 8 5 2 4 7 \quad 5 2 8 6 2 3 5 7 4 3 6 8 1 4 1 7 \quad
 5 1 8 1 3 6 5 7 3 4 2 8 6 2 4 7 \\
 2 3 8 2 4 3 6 7 5 4 1 8 1 6 5 7 \quad 3 1 8 1 3 4 6 7 5 2 4 8 2 6 5 7 \quad
 7 4 6 8 2 5 4 2 7 6 3 5 8 1 3 1 \\
 7 3 6 8 1 3 1 5 7 6 4 2 8 5 2 4 \quad 7 5 2 8 6 2 3 5 7 4 3 6 8 1 4 1 \quad
 7 5 1 8 1 3 6 5 7 3 4 2 8 6 2 4 \\
 7 2 3 8 2 4 3 6 7 5 4 1 8 1 6 5 \quad 7 3 1 8 1 3 4 6 7 5 2 4 8 2 6 5 \quad
 5 7 4 8 6 2 5 4 2 7 3 6 8 1 3 1 \\
 4 7 3 8 6 4 3 2 5 7 2 6 8 1 5 1 \quad 4 7 1 8 1 4 6 2 5 7 2 3 8 6 5 3 \quad
 5 7 2 8 3 2 5 6 3 7 4 1 8 1 6 4 \\
 4 7 5 8 1 4 1 6 5 7 2 3 8 2 6 3 \quad 4 7 3 8 5 4 3 6 2 7 5 2 8 1 6 1 \quad
 6 2 7 8 2 3 4 6 5 3 7 4 8 1 5 1 \\
 6 3 7 8 1 3 1 6 4 5 7 2 8 4 2 5 \quad 4 2 7 8 2 4 6 1 5 1 7 3 8 6 5 3 \quad
 5 3 7 8 4 3 5 6 2 4 7 2 8 1 6 1 \\
 6 2 3 8 2 7 3 6 5 1 4 1 8 7 5 4 \quad 6 4 1 8 1 7 4 6 2 5 3 2 8 7 3 5 \quad
 5 6 1 8 1 7 5 2 6 4 2 3 8 7 4 3 \\
 6 2 5 8 2 3 7 6 5 3 4 1 8 1 7 4 \quad 3 5 6 8 3 4 7 5 2 6 4 2 8 1 7 1 \quad
 4 2 6 8 2 4 7 5 1 6 1 3 8 5 7 3 \\

\noindent
$p(2,11)=17792$. The following is a partial list of all its solutions:
\medskip \\
 11 6 10 2 9 3 2 8 6 3 7 5 11 10 9 4 8 5 7 1 4 1 \quad
 11 3 10 5 9 3 4 8 6 5 7 4 11 10 9 6 8 2 7 1 2 1 \\
 11 6 10 2 9 4 2 8 6 5 4 7 11 10 9 5 8 3 1 7 1 3 \quad
 11 4 10 6 9 2 4 8 2 5 6 7 11 10 9 5 8 3 1 7 1 3 \\
 11 3 10 5 9 3 1 8 1 5 6 7 11 10 9 4 8 6 2 7 4 2 \quad
 11 3 10 4 9 3 2 8 4 2 6 7 11 10 9 5 8 6 1 7 1 5 \\
 11 4 10 6 9 1 4 1 8 5 6 7 11 10 9 5 3 8 2 7 3 2 \quad
 11 5 10 1 9 1 3 5 8 6 3 7 11 10 9 4 6 8 2 7 4 2 \\
 11 7 10 6 3 9 4 8 3 7 6 4 11 10 5 9 8 1 2 1 5 2 \quad
 11 7 10 4 6 9 3 8 4 7 3 6 11 10 5 9 8 1 2 1 5 2 \\
 11 6 10 7 5 9 2 8 6 2 5 7 11 10 4 9 8 3 1 4 1 3 \quad
 11 6 10 7 4 9 3 8 6 4 3 7 11 10 5 9 8 1 2 1 5 2 \\
 11 2 10 4 2 9 3 8 4 7 3 6 11 10 5 9 8 7 6 1 5 1 \quad
 11 6 10 2 5 9 2 8 6 4 5 7 11 10 4 9 8 3 1 7 1 3 \\
 11 7 10 3 5 9 4 3 8 7 5 4 11 10 6 9 2 8 1 2 1 6 \quad
 11 5 10 7 4 9 3 5 8 4 3 7 11 10 6 9 2 8 1 2 1 6 \\
 11 5 10 2 6 9 2 5 8 4 7 6 11 10 4 9 3 8 7 1 3 1 \quad
 11 2 10 5 2 9 4 6 8 5 7 4 11 10 6 9 3 8 7 1 3 1 \\
 11 2 10 3 2 9 4 3 8 6 7 4 11 10 5 9 6 8 7 1 5 1 \quad
 11 5 10 6 1 9 1 5 8 4 6 7 11 10 4 9 3 8 2 7 3 2 \\[5pt]
 8 11 7 2 9 5 2 10 6 8 7 5 4 11 9 6 3 4 10 1 3 1 \quad
 8 11 2 7 9 2 1 10 1 8 5 7 6 11 9 4 5 3 10 6 4 3 \\
 8 11 2 6 9 2 1 10 1 8 6 7 4 11 9 5 3 4 10 7 3 5 \quad
 8 11 2 6 9 2 1 10 1 8 6 5 7 11 9 3 4 5 10 3 7 4 \\
 8 11 3 6 9 2 3 10 2 8 6 4 7 11 9 5 4 1 10 1 7 5 \quad
 7 11 8 1 9 1 3 10 7 5 3 8 6 11 9 5 4 2 10 6 2 4 \\
 7 11 8 1 9 1 2 10 7 2 5 8 6 11 9 4 5 3 10 6 4 3 \quad
 2 11 8 2 9 3 5 10 7 3 6 8 5 11 9 4 7 6 10 1 4 1 \\
 5 11 8 6 9 2 5 10 2 7 6 8 3 11 9 4 3 7 10 1 4 1 \quad
 2 11 8 2 9 4 5 10 6 7 4 8 5 11 9 6 3 7 10 1 3 1 \\
 4 11 8 2 9 4 2 10 1 7 1 8 6 11 9 5 3 7 10 6 3 5 \quad
 2 11 8 2 9 4 1 10 1 7 4 8 6 11 9 5 3 7 10 6 3 5 \\
 4 11 8 6 9 4 1 10 1 5 6 8 7 11 9 5 2 3 10 2 7 3 \quad
 7 11 3 8 9 5 3 10 7 6 4 5 8 11 9 4 6 2 10 1 2 1 \\
 7 11 2 8 9 2 3 10 7 6 3 5 8 11 9 4 6 5 10 1 4 1 \quad
 7 11 2 8 9 2 4 10 7 5 6 4 8 11 9 5 3 6 10 1 3 1 \\
 5 11 7 8 9 3 5 10 6 3 7 4 8 11 9 6 4 2 10 1 2 1 \quad
 4 11 7 8 9 4 2 10 6 2 7 5 8 11 9 6 3 5 10 1 3 1 \\
 5 11 6 8 9 4 5 10 7 6 4 3 8 11 9 3 7 2 10 1 2 1 \quad
 5 11 2 8 9 2 5 10 1 6 1 7 8 11 9 4 6 3 10 7 4 3 \\[5pt]
 7 4 11 8 5 10 4 9 7 1 5 1 8 6 11 3 10 9 2 3 6 2 \quad
 5 7 11 8 3 10 5 9 3 7 2 6 8 2 11 4 10 9 6 1 4 1 \\
 3 4 11 8 3 10 4 9 2 5 7 2 8 6 11 5 10 9 7 1 6 1 \quad
 3 4 11 8 3 10 4 9 1 5 1 7 8 6 11 5 10 9 2 7 6 2 \\
 3 6 11 7 3 10 5 9 6 8 2 7 5 2 11 4 10 9 8 1 4 1 \quad
 5 2 11 7 2 10 5 9 3 8 4 7 3 6 11 4 10 9 8 1 6 1 \\
 2 4 11 2 3 10 4 9 3 8 1 7 1 6 11 5 10 9 8 7 6 5 \quad
 5 2 11 6 2 10 5 9 4 8 6 3 7 4 11 3 10 9 8 1 7 1 \\
 5 7 11 8 6 10 5 2 9 7 2 6 8 3 11 4 10 3 9 1 4 1 \quad
 3 7 11 8 3 10 6 2 9 7 2 5 8 6 11 4 10 5 9 1 4 1 \\
 3 7 11 8 3 10 4 2 9 7 2 4 8 6 11 5 10 1 9 1 6 5 \quad
 5 2 11 8 2 10 5 3 9 7 4 3 8 6 11 4 10 7 9 1 6 1 \\
 5 2 11 8 2 10 5 1 9 1 4 7 8 6 11 4 10 3 9 7 6 3 \quad
 6 4 11 8 5 10 4 6 9 1 5 1 8 7 11 3 10 2 9 3 2 7 \\
 5 7 11 6 8 10 5 4 9 7 6 3 4 8 11 3 10 2 9 1 2 1 \quad
 2 7 11 2 8 10 5 3 9 7 6 3 5 8 11 4 10 6 9 1 4 1 \\
 2 4 11 2 8 10 4 3 9 5 6 3 7 8 11 5 10 6 9 1 7 1 \quad
 2 7 11 2 5 10 8 3 9 7 5 3 4 6 11 8 10 4 9 1 6 1 \\
 2 5 11 2 4 10 8 5 9 4 1 7 1 6 11 8 10 3 9 7 6 3 \quad
 7 4 11 8 6 10 4 2 7 9 2 6 8 3 11 5 10 3 1 9 1 5 \\[5pt]
 7 5 9 11 8 10 3 5 7 4 3 6 9 8 4 11 10 2 6 1 2 1 \quad
 5 6 9 11 8 10 5 3 6 4 7 3 9 8 4 11 10 2 7 1 2 1 \\
 5 6 9 11 8 10 5 2 6 4 2 7 9 8 4 11 10 3 1 7 1 3 \quad
 3 7 9 11 3 10 4 2 8 7 2 4 9 5 6 11 10 8 1 5 1 6 \\
 6 2 9 11 2 10 4 6 8 3 7 4 9 3 5 11 10 8 7 1 5 1 \quad
 5 2 9 11 2 10 5 3 8 4 7 3 9 6 4 11 10 8 7 1 6 1 \\
 7 4 9 11 6 10 4 5 7 8 3 6 9 5 3 11 10 2 8 1 2 1 \quad
 7 4 9 11 6 10 4 2 7 8 2 6 9 3 5 11 10 3 8 1 5 1 \\
 7 4 9 11 5 10 4 3 7 8 5 3 9 6 2 11 10 2 8 1 6 1 \quad
 7 4 9 11 2 10 4 2 7 8 3 5 9 6 3 11 10 5 8 1 6 1 \\
 3 4 9 11 3 10 4 5 1 8 1 7 9 5 6 11 10 2 8 7 2 6 \quad
 3 6 9 11 3 10 4 2 6 8 2 4 9 7 5 11 10 1 8 1 5 7 \\
 7 4 9 11 6 10 4 1 7 1 8 6 9 3 5 11 10 3 2 8 5 2 \quad
 7 5 9 11 4 10 6 5 7 4 8 2 9 6 2 11 10 3 1 8 1 3 \\
 7 4 9 11 2 10 4 2 7 3 8 5 9 3 6 11 10 5 1 8 1 6 \quad
 6 7 9 11 1 10 1 6 2 7 8 2 9 4 5 11 10 3 4 8 5 3 \\
 6 4 9 11 3 10 4 6 3 7 8 1 9 1 5 11 10 7 2 8 5 2 \quad
 5 2 9 11 2 10 5 3 4 7 8 3 9 4 6 11 10 7 1 8 1 6 \\
 3 4 9 11 3 10 4 5 2 7 8 2 9 5 6 11 10 7 1 8 1 6 \quad
 5 6 9 11 4 10 5 3 6 4 8 3 9 7 2 11 10 2 1 8 1 7 \\[5pt]
 1 5 1 7 11 6 8 5 4 9 10 7 6 4 3 8 11 2 3 9 2 10 \quad
 1 3 1 7 11 3 8 6 4 9 10 7 5 4 6 8 11 2 5 9 2 10 \\
 1 5 1 7 11 2 8 5 2 9 10 7 4 6 3 8 11 4 3 9 6 10 \quad
 6 1 2 1 11 2 8 6 4 9 10 5 7 4 3 8 11 5 3 9 7 10 \\
 1 6 1 3 11 5 8 3 6 9 10 5 7 4 2 8 11 2 4 9 7 10 \quad
 1 3 1 4 11 3 8 5 4 9 10 6 7 5 2 8 11 2 6 9 7 10 \\
 1 3 1 7 11 3 6 4 8 9 10 7 4 6 5 2 11 8 2 9 5 10 \quad
 5 3 4 7 11 3 5 4 8 9 10 7 1 6 1 2 11 8 2 9 6 10 \\
 4 5 3 7 11 4 3 5 8 9 10 7 1 6 1 2 11 8 2 9 6 10 \quad
 6 1 3 1 11 7 3 6 8 9 10 2 5 7 2 4 11 8 5 9 4 10 \\
 2 5 3 2 11 7 3 5 8 9 10 6 1 7 1 4 11 8 6 9 4 10 \quad
 1 4 1 3 11 7 4 3 8 9 10 6 2 7 5 2 11 8 6 9 5 10 \\
 4 1 3 1 11 4 3 7 8 9 10 2 5 6 2 7 11 8 5 9 6 10 \quad
 5 1 2 1 11 2 5 6 8 9 10 3 7 4 6 3 11 8 4 9 7 10 \\
 6 1 3 1 11 7 3 6 5 9 10 8 4 7 5 2 11 4 2 9 8 10 \quad
 2 6 4 2 11 7 5 4 6 9 10 8 5 7 3 1 11 1 3 9 8 10 \\
 1 5 1 2 11 7 2 5 6 9 10 8 4 7 3 6 11 4 3 9 8 10 \quad
 1 4 1 3 11 7 4 3 6 9 10 8 5 7 2 6 11 2 5 9 8 10 \\
 1 5 1 4 11 6 7 5 4 9 10 8 6 3 7 2 11 3 2 9 8 10 \quad
 1 3 1 4 11 3 6 7 4 9 10 8 5 6 2 7 11 2 5 9 8 10 \\ 

\DeclareFixedFont{\smallrm}{OT1}{cmr}{m}{n}{9pt}
\renewcommand{\quad}{\hspace{21pt}}
\noindent
$p(2,12)=108144$. The following is a partial list of all its solutions:
\medskip \\ \smallrm
 12 10 11 6 4 5 9 7 8 4 6 5 10 12 11 7 9 8 3 1 2 1 3 2 \quad
 12 10 11 6 4 5 9 7 8 4 6 5 10 12 11 7 9 8 2 3 1 2 1 3 \\
 12 10 11 5 6 4 9 7 8 5 4 6 10 12 11 7 9 8 3 1 2 1 3 2 \quad
 12 10 11 5 6 4 9 7 8 5 4 6 10 12 11 7 9 8 2 3 1 2 1 3 \\
 12 10 11 3 4 5 9 3 8 4 7 5 10 12 11 6 9 8 7 1 2 1 6 2 \quad
 12 10 11 6 4 1 9 1 8 4 6 7 10 12 11 5 9 8 3 7 2 5 3 2 \\
 12 10 11 6 2 3 9 2 8 3 6 7 10 12 11 5 9 8 4 7 1 5 1 4 \quad
 12 10 11 6 3 5 9 7 3 8 6 5 10 12 11 7 9 4 8 1 2 1 4 2 \\
 12 10 11 6 3 1 9 1 3 8 6 7 10 12 11 5 9 4 8 7 2 5 4 2 \quad
 12 10 11 4 5 1 9 1 4 8 5 7 10 12 11 6 9 3 8 7 2 3 6 2 \\
 12 10 11 6 4 5 8 9 7 4 6 5 10 12 11 8 7 9 3 1 2 1 3 2 \quad
 12 10 11 6 4 5 8 9 7 4 6 5 10 12 11 8 7 9 2 3 1 2 1 3 \\
 12 10 11 5 6 4 8 9 7 5 4 6 10 12 11 8 7 9 3 1 2 1 3 2 \quad
 12 10 11 5 6 4 8 9 7 5 4 6 10 12 11 8 7 9 2 3 1 2 1 3 \\
 12 10 11 5 6 2 8 9 2 5 7 6 10 12 11 8 4 9 7 3 1 4 1 3 \quad
 12 10 11 4 6 3 8 9 4 3 7 6 10 12 11 8 5 9 7 1 2 1 5 2 \\
 12 10 11 5 3 4 8 9 3 5 4 7 10 12 11 8 6 9 2 7 1 2 1 6 \quad
 12 10 11 4 5 3 8 9 4 3 5 7 10 12 11 8 6 9 2 7 1 2 1 6 \\
 12 10 11 2 6 4 2 9 7 8 4 6 10 12 11 5 7 9 8 3 1 5 1 3 \quad
 12 10 11 2 3 4 2 9 3 8 4 7 10 12 11 5 6 9 8 7 1 5 1 6 \\[5pt]
 5 12 8 6 10 7 5 11 9 2 6 8 2 7 12 10 3 4 9 11 3 1 4 1 \quad
 5 12 8 6 10 3 5 11 9 3 6 8 7 2 12 10 2 4 9 11 7 1 4 1 \\
 4 12 8 6 10 4 2 11 9 2 6 8 7 3 12 10 5 3 9 11 7 1 5 1 \quad
 2 12 8 2 10 3 6 11 9 3 5 8 7 6 12 10 5 4 9 11 7 1 4 1 \\
 4 12 8 2 10 4 2 11 9 3 6 8 7 3 12 10 5 6 9 11 7 1 5 1 \quad
 4 12 8 6 10 4 2 11 9 2 6 8 3 7 12 10 3 5 9 11 1 7 1 5 \\
 4 12 7 8 10 4 6 11 9 1 7 1 8 6 12 10 5 3 9 11 2 3 5 2 \quad
 1 12 1 8 10 5 3 11 9 6 3 5 8 7 12 10 6 4 9 11 2 7 4 2 \\
 4 12 3 8 10 4 3 11 9 1 6 1 8 7 12 10 5 6 9 11 2 7 5 2 \quad
 1 12 1 8 10 3 4 11 9 3 6 4 8 7 12 10 5 6 9 11 2 7 5 2 \\
 7 12 8 1 10 1 6 11 7 2 9 8 2 6 12 10 5 3 4 11 9 3 5 4 \quad
 7 12 8 1 10 1 6 11 7 2 9 8 2 6 12 10 4 5 3 11 9 4 3 5 \\
 7 12 8 1 10 1 5 11 7 3 9 8 5 3 12 10 6 4 2 11 9 2 4 6 \quad
 2 12 8 2 10 3 6 11 7 3 9 8 5 6 12 10 7 4 5 11 9 1 4 1 \\
 5 12 8 4 10 6 5 11 4 7 9 8 6 2 12 10 2 7 3 11 9 1 3 1 \quad
 2 12 8 2 10 3 1 11 1 3 9 8 5 7 12 10 6 4 5 11 9 7 4 6 \\
 7 12 2 8 10 2 6 11 7 1 9 1 8 6 12 10 5 3 4 11 9 3 5 4 \quad
 7 12 2 8 10 2 6 11 7 1 9 1 8 6 12 10 4 5 3 11 9 4 3 5 \\
 5 12 2 8 10 2 5 11 7 3 9 6 8 3 12 10 7 4 6 11 9 1 4 1 \quad
 1 12 1 8 10 2 4 11 2 6 9 4 8 7 12 10 6 5 3 11 9 7 3 5 \\[5pt]
 5 2 12 7 2 10 5 9 3 8 11 7 3 6 4 12 10 9 8 4 6 1 11 1 \quad
 3 1 12 1 3 10 5 9 4 8 11 7 5 4 6 12 10 9 8 7 2 6 11 2 \\
 5 2 12 4 2 10 5 9 4 8 11 1 7 1 6 12 10 9 8 3 7 6 11 3 \quad
 5 1 12 1 6 10 5 9 3 8 11 6 3 7 4 12 10 9 8 4 2 7 11 2 \\
 3 1 12 1 3 10 4 9 5 8 11 4 6 7 5 12 10 9 8 6 2 7 11 2 \quad
 5 2 12 7 2 10 5 4 8 9 11 7 4 6 3 12 10 8 3 9 6 1 11 1 \\
 2 5 12 2 7 10 4 5 8 9 11 4 7 6 3 12 10 8 3 9 6 1 11 1 \quad
 5 1 12 1 7 10 5 3 8 9 11 3 7 4 6 12 10 8 4 9 2 6 11 2 \\
 2 4 12 2 6 10 4 5 8 9 11 6 7 5 3 12 10 8 3 9 7 1 11 1 \quad
 5 2 12 3 2 10 5 3 8 9 11 1 7 1 6 12 10 8 4 9 7 6 11 4 \\
 3 5 12 2 3 10 2 5 8 9 11 1 7 1 6 12 10 8 4 9 7 6 11 4 \quad
 5 1 12 1 6 10 5 4 8 9 11 6 4 7 3 12 10 8 3 9 2 7 11 2 \\
 3 1 12 1 3 10 4 6 8 9 11 4 5 7 6 12 10 8 5 9 2 7 11 2 \quad
 6 1 12 1 7 10 2 6 8 2 11 9 7 4 5 12 10 8 4 3 5 9 11 3 \\
 2 5 12 2 7 10 6 5 8 4 11 9 7 6 4 12 10 8 1 3 1 9 11 3 \quad
 3 1 12 1 3 10 2 6 8 2 11 9 7 5 6 12 10 8 4 5 7 9 11 4 \\
 7 1 12 1 2 10 6 2 7 8 11 9 4 6 5 12 10 4 8 3 5 9 11 3 \quad
 5 1 12 1 7 10 5 6 4 8 11 9 7 4 6 12 10 2 8 3 2 9 11 3 \\
 6 1 12 1 3 10 7 6 3 8 11 9 4 5 7 12 10 4 8 5 2 9 11 2 \quad
 2 7 12 2 6 8 10 4 9 7 11 6 4 5 8 12 3 10 9 5 3 1 11 1 \\[5pt]
 10 5 3 12 9 6 3 5 11 8 2 10 6 2 9 7 12 4 8 1 11 1 4 7 \quad
 10 7 3 12 9 6 3 5 11 7 8 10 6 5 9 2 12 4 2 8 11 1 4 1 \\
 10 4 6 12 9 1 4 1 11 6 8 10 2 7 9 2 12 5 3 8 11 7 3 5 \quad
 10 7 2 12 9 2 3 5 11 7 3 10 8 5 9 6 12 1 4 1 11 8 6 4 \\
 10 5 7 12 9 3 6 5 11 3 7 10 8 6 9 1 12 1 4 2 11 8 2 4 \quad
 10 3 7 12 9 3 2 5 11 2 7 10 8 5 9 6 12 1 4 1 11 8 6 4 \\
 10 5 6 12 9 7 3 5 11 6 3 10 8 7 9 1 12 1 4 2 11 8 2 4 \quad
 10 4 2 12 9 2 4 1 11 1 6 10 8 5 9 7 12 6 3 5 11 8 3 7 \\
 10 7 3 12 9 6 3 2 11 7 2 10 6 8 9 5 12 1 4 1 11 5 8 4 \quad
 10 7 4 12 9 6 3 4 11 7 3 10 6 8 9 2 12 5 2 1 11 1 8 5 \\
 10 7 3 12 9 1 3 1 11 7 4 10 6 8 9 4 12 5 2 6 11 2 8 5 \quad
 10 4 7 12 9 6 4 1 11 1 7 10 6 8 9 3 12 5 2 3 11 2 8 5 \\
 10 5 3 12 9 4 3 5 11 7 4 10 6 8 9 2 12 7 2 6 11 1 8 1 \quad
 10 3 4 12 9 3 2 4 11 2 6 10 5 8 9 7 12 6 5 1 11 1 8 7 \\
 10 2 3 12 2 9 3 4 11 8 6 10 4 5 7 9 12 6 8 5 11 1 7 1 \quad
 10 4 6 12 7 9 4 5 11 6 8 10 7 5 2 9 12 2 3 8 11 1 3 1 \\
 10 2 6 12 2 9 7 5 11 6 8 10 4 5 7 9 12 4 3 8 11 1 3 1 \quad
 10 4 6 12 2 9 4 2 11 6 8 10 1 7 1 9 12 5 3 8 11 7 3 5 \\
 10 4 1 12 1 9 4 6 11 2 8 10 2 7 6 9 12 5 3 8 11 7 3 5 \quad
 10 4 6 12 2 9 4 2 11 6 8 10 5 3 7 9 12 3 5 8 11 1 7 1 \\[5pt]
 7 11 3 5 12 9 3 8 7 5 10 2 6 11 2 9 8 12 4 6 1 10 1 4 \quad
 5 11 6 4 12 9 5 8 4 6 10 7 1 11 1 9 8 12 3 7 2 10 3 2 \\
 1 11 1 4 12 9 7 3 4 8 10 3 6 11 7 9 5 12 8 6 2 10 5 2 \quad
 6 11 3 4 12 9 3 6 4 8 10 7 1 11 1 9 5 12 8 7 2 10 5 2 \\
 2 11 3 2 12 9 3 4 5 8 10 7 4 11 5 9 6 12 8 7 1 10 1 6 \quad
 1 11 1 3 12 9 4 3 5 8 10 4 7 11 5 9 6 12 8 2 7 10 2 6 \\
 7 11 6 3 12 9 5 3 7 6 10 8 5 11 1 9 1 12 4 2 8 10 2 4 \quad
 6 11 3 5 12 9 3 6 7 5 10 8 1 11 1 9 7 12 4 2 8 10 2 4 \\
 5 11 6 3 12 9 5 3 7 6 10 8 1 11 1 9 7 12 4 2 8 10 2 4 \quad
 2 11 4 2 12 9 1 4 1 6 10 8 5 11 7 9 6 12 5 3 8 10 7 3 \\
 7 11 4 5 12 8 9 4 7 5 10 6 2 11 8 2 9 12 6 3 1 10 1 3 \quad
 6 11 5 7 12 8 9 6 5 2 10 7 2 11 8 3 9 12 4 3 1 10 1 4 \\
 1 11 1 5 12 8 9 4 6 5 10 7 4 11 8 6 9 12 3 7 2 10 3 2 \quad
 1 11 1 5 12 8 9 3 6 5 10 3 7 11 8 6 9 12 4 2 7 10 2 4 \\
 4 11 5 7 12 4 9 6 5 8 10 7 2 11 6 2 9 12 8 3 1 10 1 3 \quad
 1 11 1 7 12 2 9 4 2 8 10 7 4 11 5 6 9 12 8 3 5 10 6 3 \\
 4 11 5 3 12 4 9 3 5 8 10 7 1 11 1 6 9 12 8 7 2 10 6 2 \quad
 2 11 5 2 12 6 9 3 5 8 10 3 6 11 4 7 9 12 8 4 1 10 1 7 \\
 1 11 1 4 12 6 9 3 4 8 10 3 6 11 5 7 9 12 8 2 5 10 2 7 \quad
 4 11 5 3 12 4 9 3 5 8 10 2 6 11 2 7 9 12 8 6 1 10 1 7 \\[5pt]
 7 4 1 9 1 12 4 8 7 10 11 5 2 9 6 2 8 5 12 3 10 6 11 3 \quad
 7 4 1 9 1 12 4 6 7 10 11 3 8 9 6 3 2 5 12 2 10 8 11 5 \\
 7 5 1 9 1 12 4 5 7 10 11 4 8 9 2 3 6 2 12 3 10 8 11 6 \quad
 6 2 5 9 2 12 7 6 5 10 11 4 8 9 7 1 4 1 12 3 10 8 11 3 \\
 5 6 1 9 1 12 5 7 6 10 11 4 8 9 2 7 4 2 12 3 10 8 11 3 \quad
 6 4 1 9 1 12 4 6 7 10 11 2 8 9 2 3 7 5 12 3 10 8 11 5 \\
 7 1 4 1 9 12 6 4 7 10 11 3 8 6 9 3 2 5 12 2 10 8 11 5 \quad
 6 4 5 7 9 12 4 6 5 10 11 7 8 2 9 1 2 1 12 3 10 8 11 3 \\
 5 6 4 7 9 12 5 4 6 10 11 7 8 2 9 1 2 1 12 3 10 8 11 3 \quad
 7 2 4 8 2 12 9 4 7 10 11 5 8 1 6 1 9 5 12 3 10 6 11 3 \\
 7 2 4 8 2 12 9 4 7 10 11 1 8 1 6 3 9 5 12 3 10 6 11 5 \quad
 7 1 4 1 8 12 9 4 7 10 11 5 2 8 6 2 9 5 12 3 10 6 11 3 \\
 7 2 3 8 2 12 3 9 7 10 11 1 8 1 4 5 6 9 12 4 10 5 11 6 \quad
 7 2 4 8 2 12 6 4 7 10 11 9 8 6 1 3 1 5 12 3 10 9 11 5 \\
 7 4 1 8 1 12 4 6 7 10 11 9 8 2 6 3 2 5 12 3 10 9 11 5 \quad
 7 4 1 8 1 12 4 5 7 10 11 9 8 5 2 3 6 2 12 3 10 9 11 6 \\
 7 2 3 8 2 12 3 5 7 10 11 9 8 5 4 1 6 1 12 4 10 9 11 6 \quad
 6 2 3 8 2 12 3 6 7 10 11 9 8 1 4 1 7 5 12 4 10 9 11 5 \\
 7 1 4 1 8 12 5 4 7 10 11 9 5 8 2 3 6 2 12 3 10 9 11 6 \quad
 6 1 3 1 8 12 3 6 7 10 11 9 2 8 4 2 7 5 12 4 10 9 11 5 \\

\normalsize
\noindent
$p(2,15)=39809640$. The following is a partial list of all its solutions:
\medskip \\
 15 13 14 8 5 12 7 11 4 10 5 9 8 4 7 13 15 14 12 11 10 9 6 3 1 2 1 3 2 6 \\
 15 13 14 8 5 12 7 11 4 10 5 9 8 4 7 13 15 14 12 11 10 9 6 2 3 1 2 1 3 6 \\
 15 13 14 6 8 12 7 11 3 10 6 9 3 8 7 13 15 14 12 11 10 9 5 2 4 1 2 1 5 4 \\
 15 13 14 6 8 12 7 11 3 10 6 9 3 8 7 13 15 14 12 11 10 9 4 5 1 2 1 4 2 5 \\
 15 13 14 7 8 12 3 11 5 10 3 7 9 8 5 13 15 14 12 11 10 6 9 2 4 1 2 1 6 4 \\
 15 13 14 6 8 12 7 11 2 10 6 2 9 8 7 13 15 14 12 11 10 5 9 4 1 3 1 5 4 3 \\
 15 13 14 6 8 12 7 11 2 10 6 2 9 8 7 13 15 14 12 11 10 4 9 5 3 1 4 1 3 5 \\
 15 13 14 6 8 12 4 11 5 10 6 4 9 8 5 13 15 14 12 11 10 7 9 3 1 2 1 3 2 7 \\
 15 13 14 6 8 12 4 11 5 10 6 4 9 8 5 13 15 14 12 11 10 7 9 2 3 1 2 1 3 7 \\
 15 13 14 8 6 12 7 11 1 10 1 6 8 9 7 13 15 14 12 11 10 5 3 9 4 2 3 5 2 4 \\[5pt]
 9 15 8 1 14 1 11 6 13 5 9 8 12 10 6 5 4 15 11 14 7 4 13 3 10 12 2 3 7 2 \\
 9 15 4 7 14 8 11 4 13 5 9 7 12 10 8 5 3 15 11 14 3 6 13 2 10 12 2 1 6 1 \\
 9 15 6 3 14 8 11 3 13 6 9 4 12 10 8 5 4 15 11 14 7 5 13 2 10 12 2 1 7 1 \\
 9 15 4 1 14 1 11 4 13 6 9 3 12 10 8 3 6 15 11 14 7 5 13 8 10 12 2 5 7 2 \\
 9 15 2 3 14 2 11 3 13 5 9 4 12 10 8 5 4 15 11 14 6 7 13 8 10 12 1 6 1 7 \\
 9 15 2 7 14 2 11 6 13 5 9 7 12 10 6 5 4 15 11 14 8 4 13 3 10 12 1 3 1 8 \\
 9 15 2 5 14 2 11 7 13 5 9 4 12 10 6 7 4 15 11 14 8 6 13 3 10 12 1 3 1 8 \\
 8 15 7 3 14 9 11 3 13 8 7 2 12 10 2 9 4 15 11 14 6 4 13 5 10 12 1 6 1 5 \\
 8 15 7 3 14 9 11 3 13 8 7 4 12 10 5 9 4 15 11 14 5 6 13 2 10 12 2 1 6 1 \\
 6 15 8 5 14 9 11 6 13 5 3 8 12 10 3 9 4 15 11 14 7 4 13 2 10 12 2 1 7 1 \\[5pt]
 9 7 15 13 10 5 6 8 12 7 9 5 14 6 11 10 8 13 15 3 1 12 1 3 4 2 11 14 2 4 \\
 9 4 15 13 10 7 4 8 12 1 9 1 14 7 11 10 8 13 15 5 2 12 6 2 3 5 11 14 3 6 \\
 9 3 15 13 10 3 2 8 12 2 9 7 14 6 11 10 8 13 15 7 6 12 5 1 4 1 11 14 5 4 \\
 9 4 15 13 10 3 4 8 12 3 9 7 14 6 11 10 8 13 15 7 6 12 1 5 1 2 11 14 2 5 \\
 9 3 15 13 10 3 6 8 12 1 9 1 14 6 11 10 8 13 15 4 7 12 2 5 4 2 11 14 7 5 \\
 9 7 15 13 10 1 6 1 12 7 9 8 14 6 11 10 5 13 15 3 8 12 5 3 4 2 11 14 2 4 \\
 9 4 15 13 10 1 4 1 12 6 9 7 14 8 11 10 6 13 15 7 3 12 8 5 3 2 11 14 2 5 \\
 9 7 15 13 10 1 4 1 12 7 9 4 14 6 11 10 8 13 15 2 6 12 2 5 3 8 11 14 3 5 \\
 9 7 15 13 10 1 4 1 12 7 9 4 14 5 11 10 8 13 15 5 2 12 6 2 3 8 11 14 3 6 \\
 9 3 15 13 10 3 6 1 12 1 9 7 14 6 11 10 8 13 15 7 2 12 5 2 4 8 11 14 5 4 \\[5pt]
 7 2 9 15 2 8 14 10 7 12 5 11 9 13 8 4 5 6 10 15 4 14 12 11 6 3 1 13 1 3 \\
 4 2 9 15 2 4 14 10 7 12 6 11 9 13 5 8 7 6 10 15 5 14 12 11 8 3 1 13 1 3 \\
 1 3 1 15 9 3 14 10 6 12 5 11 7 13 9 6 5 8 10 15 7 14 12 11 4 2 8 13 2 4 \\
 1 3 1 15 5 3 14 10 6 12 5 11 4 13 9 6 8 4 10 15 7 14 12 11 9 8 2 13 7 2 \\
 4 6 1 15 1 4 14 10 6 12 5 11 2 13 9 2 5 8 10 15 7 14 12 11 9 3 8 13 7 3 \\
 4 8 1 15 1 4 14 10 5 12 8 11 3 13 5 9 3 6 10 15 7 14 12 11 6 9 2 13 7 2 \\
 3 6 2 15 3 2 14 10 6 12 5 11 1 13 1 9 5 8 10 15 7 14 12 11 4 9 8 13 7 4 \\
 4 2 7 15 2 4 14 10 5 12 7 11 3 13 5 9 3 6 10 15 8 14 12 11 6 9 1 13 1 8 \\
 4 8 1 15 1 4 14 10 5 12 8 11 2 13 5 2 9 6 10 15 7 14 12 11 6 3 9 13 7 3 \\
 4 2 5 15 2 4 14 10 5 12 6 11 1 13 1 8 9 6 10 15 7 14 12 11 8 3 9 13 7 3 \\[5pt]
 8 9 3 14 15 11 3 7 12 8 13 9 6 2 10 7 2 11 14 6 15 12 4 5 13 10 1 4 1 5 \\
 8 9 3 14 15 11 3 4 12 8 13 9 4 2 10 7 2 11 14 6 15 12 5 7 13 10 6 1 5 1 \\
 8 9 3 14 15 11 3 4 12 8 13 9 4 2 10 6 2 11 14 7 15 12 6 5 13 10 1 7 1 5 \\
 8 9 3 14 15 11 3 4 12 8 13 9 4 1 10 1 6 11 14 7 15 12 5 6 13 10 2 7 5 2 \\
 1 9 1 14 15 11 8 4 12 3 13 9 4 3 10 8 6 11 14 7 15 12 5 6 13 10 2 7 5 2 \\
 1 9 1 14 15 11 5 8 12 3 13 9 5 3 10 6 8 11 14 7 15 12 6 4 13 10 2 7 4 2 \\
 1 9 1 14 15 11 5 7 12 3 13 9 5 3 10 7 6 11 14 8 15 12 4 6 13 10 2 4 8 2 \\
 5 9 6 14 15 11 5 4 12 6 13 9 4 2 10 7 2 11 14 8 15 12 3 7 13 10 3 1 8 1 \\
 8 4 9 14 15 11 4 3 12 8 13 3 9 2 10 7 2 11 14 6 15 12 5 7 13 10 6 1 5 1 \\
 8 4 9 14 15 11 4 3 12 8 13 3 9 2 10 6 2 11 14 7 15 12 6 5 13 10 1 7 1 5 \\[5pt]
 14 9 6 8 3 15 10 12 3 6 11 9 8 13 1 14 1 10 7 4 12 15 11 5 4 2 7 13 2 5 \\
 14 9 1 8 1 15 10 12 5 6 11 9 8 13 5 14 6 10 7 3 12 15 11 3 4 2 7 13 2 4 \\
 14 9 6 4 7 15 10 12 4 6 11 9 7 13 8 14 2 10 5 2 12 15 11 8 5 3 1 13 1 3 \\
 14 9 6 4 7 15 10 12 4 6 11 9 7 13 8 14 1 10 1 5 12 15 11 8 3 5 2 13 3 2 \\
 14 9 6 4 7 15 10 12 4 6 11 9 7 13 1 14 1 10 5 8 12 15 11 2 5 3 2 13 8 3 \\
 14 2 9 7 2 15 10 12 5 8 11 7 9 13 5 14 6 10 8 4 12 15 11 6 4 3 1 13 1 3 \\
 14 7 9 4 6 15 10 12 4 7 11 6 9 13 8 14 2 10 5 2 12 15 11 8 5 3 1 13 1 3 \\
 14 7 9 4 6 15 10 12 4 7 11 6 9 13 8 14 1 10 1 5 12 15 11 8 3 5 2 13 3 2 \\
 14 6 9 7 4 15 10 12 6 4 11 7 9 13 8 14 2 10 5 2 12 15 11 8 5 3 1 13 1 3 \\
 14 6 9 7 4 15 10 12 6 4 11 7 9 13 8 14 1 10 1 5 12 15 11 8 3 5 2 13 3 2 \\[5pt]
 3 1 4 1 3 5 15 4 7 9 10 5 11 13 14 12 7 8 6 9 2 10 15 2 11 6 8 13 12 14 \\
 8 5 2 4 6 2 15 5 4 8 10 6 11 13 14 12 9 7 1 3 1 10 15 3 11 7 9 13 12 14 \\
 7 5 8 1 4 1 15 5 7 4 10 8 11 13 14 12 9 2 6 3 2 10 15 3 11 6 9 13 12 14 \\
 6 7 8 1 3 1 15 6 3 7 10 8 11 13 14 12 9 2 4 5 2 10 15 4 11 5 9 13 12 14 \\
 4 6 8 5 2 4 15 2 6 5 10 8 11 13 14 12 9 7 1 3 1 10 15 3 11 7 9 13 12 14 \\
 4 7 1 3 1 4 15 3 8 7 10 6 11 13 14 12 9 8 6 5 2 10 15 2 11 5 9 13 12 14 \\
 3 7 4 1 3 1 15 4 8 7 10 6 11 13 14 12 9 8 6 5 2 10 15 2 11 5 9 13 12 14 \\
 6 1 2 1 3 2 15 6 3 8 10 7 11 13 14 12 9 5 8 7 4 10 15 5 11 4 9 13 12 14 \\
 6 1 2 1 3 2 15 6 3 7 10 8 11 13 14 12 9 7 4 5 8 10 15 4 11 5 9 13 12 14 \\
 4 1 5 1 2 4 15 2 5 7 10 8 11 13 14 12 9 7 6 3 8 10 15 3 11 6 9 13 12 14 \\

\normalsize
\noindent
$p(2,16)=326721800$. The following is a partial list of all its solutions:
\medskip \\
 16 14 15 9 7 13 3 12 6 11 3 10 7 9 8 6 14 16 15 13 12 11 10 8 5 2 4 1 2 1 5 4 \\
 16 14 15 9 7 13 3 12 6 11 3 10 7 9 8 6 14 16 15 13 12 11 10 8 4 5 1 2 1 4 2 5 \\
 16 14 15 9 7 13 3 12 5 11 3 10 7 9 5 8 14 16 15 13 12 11 10 6 8 2 4 1 2 1 6 4 \\
 16 14 15 9 5 13 7 12 3 11 5 10 3 9 7 8 14 16 15 13 12 11 10 6 8 2 4 1 2 1 6 4 \\
 16 14 15 6 9 13 5 12 4 11 6 10 5 4 9 8 14 16 15 13 12 11 10 7 8 3 1 2 1 3 2 7 \\
 16 14 15 6 9 13 5 12 4 11 6 10 5 4 9 8 14 16 15 13 12 11 10 7 8 2 3 1 2 1 3 7 \\
 16 14 15 8 5 13 7 12 6 11 5 10 8 9 7 6 14 16 15 13 12 11 10 9 4 1 3 1 2 4 3 2 \\
 16 14 15 8 5 13 7 12 6 11 5 10 8 9 7 6 14 16 15 13 12 11 10 9 2 3 4 2 1 3 1 4 \\
 16 14 15 6 7 13 8 12 5 11 6 10 7 9 5 8 14 16 15 13 12 11 10 9 4 1 3 1 2 4 3 2 \\
 16 14 15 6 7 13 8 12 5 11 6 10 7 9 5 8 14 16 15 13 12 11 10 9 2 3 4 2 1 3 1 4 \\[5pt]
 10 16 5 12 9 11 15 1 5 1 8 10 14 7 9 13 12 11 16 8 3 7 15 6 3 2 4 14 2 13 6 4 \\
 10 16 6 12 9 11 15 5 1 6 1 10 14 5 9 13 12 11 16 8 4 2 15 7 2 4 3 14 8 13 3 7 \\
 10 16 7 12 9 11 15 1 4 1 7 10 14 4 9 13 12 11 16 6 2 8 15 2 3 5 6 14 3 13 8 5 \\
 10 16 7 12 9 11 15 1 4 1 7 10 14 4 9 13 12 11 16 3 5 8 15 3 6 2 5 14 2 13 8 6 \\
 10 16 6 12 9 11 15 2 4 6 2 10 14 4 9 13 12 11 16 3 7 8 15 3 1 5 1 14 7 13 8 5 \\
 10 16 4 12 9 11 15 4 1 3 1 10 14 3 9 13 12 11 16 5 7 8 15 2 6 5 2 14 7 13 8 6 \\
 10 16 6 12 9 11 15 2 4 6 2 10 14 4 9 13 12 11 16 5 3 8 15 7 3 5 1 14 1 13 8 7 \\
 10 16 5 12 9 11 15 1 5 1 2 10 14 2 9 13 12 11 16 6 4 8 15 7 3 4 6 14 3 13 8 7 \\
 10 16 5 12 8 11 15 2 5 9 2 10 14 8 6 13 12 11 16 9 3 6 15 7 3 4 1 14 1 13 4 7 \\
 10 16 7 12 2 11 15 2 4 9 7 10 14 4 8 13 12 11 16 9 5 3 15 8 6 3 5 14 1 13 1 6 \\[5pt]
 10 8 16 4 11 14 5 9 4 13 8 10 5 15 7 12 11 9 2 16 14 2 7 13 6 1 3 1 12 15 3 6 \\
 10 8 16 2 11 14 2 9 3 13 8 10 3 15 6 12 11 9 5 16 14 6 7 13 5 1 4 1 12 15 7 4 \\
 10 6 16 8 11 14 3 9 6 13 3 10 8 15 7 12 11 9 2 16 14 2 7 13 5 1 4 1 12 15 5 4 \\
 10 4 16 8 11 14 4 9 1 13 1 10 8 15 6 12 11 9 7 16 14 6 2 13 5 2 7 3 12 15 5 3 \\
 10 6 16 2 11 14 2 9 6 13 3 10 8 15 3 12 11 9 5 16 14 8 7 13 5 1 4 1 12 15 7 4 \\
 10 1 16 1 11 14 3 9 5 13 3 10 8 15 5 12 11 9 6 16 14 8 7 13 2 6 4 2 12 15 7 4 \\
 10 6 16 2 11 14 2 9 6 13 1 10 1 15 7 12 11 9 8 16 14 3 7 13 5 3 4 8 12 15 5 4 \\
 10 8 16 3 11 14 5 3 9 13 8 10 5 15 6 12 11 7 9 16 14 6 1 13 1 7 4 2 12 15 2 4 \\
 10 8 16 3 11 14 5 3 9 13 8 10 5 15 2 12 11 2 9 16 14 4 7 13 6 1 4 1 12 15 7 6 \\
 10 5 16 8 11 14 7 5 9 13 6 10 8 15 7 12 11 6 9 16 14 3 1 13 1 3 4 2 12 15 2 4 \\[5pt]
 9 11 6 16 2 10 5 2 15 6 9 13 5 11 14 12 10 8 4 1 16 1 7 4 15 13 8 3 12 14 7 3 \\
 9 11 6 16 2 10 5 2 15 6 9 13 5 11 14 12 10 4 8 1 16 1 4 7 15 13 3 8 12 14 3 7 \\
 9 11 6 16 1 10 1 4 15 6 9 13 4 11 14 12 10 5 7 2 16 8 2 5 15 13 7 3 12 14 8 3 \\
 9 11 6 16 1 10 1 4 15 6 9 13 4 11 14 12 10 5 2 7 16 2 8 5 15 13 3 7 12 14 3 8 \\
 5 11 9 16 2 10 5 2 15 8 6 13 9 11 14 12 10 6 8 1 16 1 7 3 15 13 4 3 12 14 7 4 \\
 5 11 9 16 2 10 5 2 15 8 6 13 9 11 14 12 10 6 8 1 16 1 4 7 15 13 3 4 12 14 3 7 \\
 6 11 9 16 5 10 2 6 15 2 5 13 9 11 14 12 10 8 4 1 16 1 7 4 15 13 8 3 12 14 7 3 \\
 6 11 9 16 5 10 2 6 15 2 5 13 9 11 14 12 10 4 8 1 16 1 4 7 15 13 3 8 12 14 3 7 \\
 6 11 9 16 4 10 3 6 15 4 3 13 9 11 14 12 10 5 7 1 16 1 8 5 15 13 7 2 12 14 2 8 \\
 8 11 7 16 1 10 1 9 15 8 7 13 5 11 14 12 10 9 5 2 16 4 2 6 15 13 4 3 12 14 6 3 \\[5pt]
 14 10 3 9 16 13 3 1 8 1 15 6 10 9 12 14 11 8 6 13 2 16 7 2 4 5 15 12 11 4 7 5 \\
 14 10 4 9 16 13 2 4 8 2 15 6 10 9 12 14 11 8 6 13 1 16 1 7 5 3 15 12 11 3 5 7 \\
 14 10 4 9 16 13 1 4 1 8 15 5 10 9 12 14 11 5 8 13 2 16 7 2 6 3 15 12 11 3 7 6 \\
 14 10 3 9 16 13 3 1 2 1 15 2 10 9 12 14 11 6 4 13 8 16 7 4 6 5 15 12 11 8 7 5 \\
 14 10 3 9 16 13 3 1 2 1 15 2 10 9 12 14 11 4 6 13 8 16 4 7 5 6 15 12 11 8 5 7 \\
 14 10 4 9 16 13 1 4 1 7 15 6 10 9 12 14 11 7 6 13 2 16 8 2 5 3 15 12 11 3 5 8 \\
 14 10 5 9 16 13 2 3 5 2 15 3 10 9 12 14 11 7 1 13 1 16 8 6 4 7 15 12 11 4 6 8 \\
 14 10 8 5 16 13 1 9 1 5 15 8 10 4 12 14 11 9 4 13 2 16 7 2 6 3 15 12 11 3 7 6 \\
 14 10 5 2 16 13 2 9 5 8 15 1 10 1 12 14 11 9 8 13 3 16 6 7 3 4 15 12 11 6 4 7 \\
 14 10 3 5 16 13 3 9 2 5 15 2 10 4 12 14 11 9 4 13 8 16 7 1 6 1 15 12 11 8 7 6 \\[5pt]
 12 9 6 15 10 16 3 8 14 6 3 9 11 12 13 10 8 2 7 15 2 5 16 14 11 4 7 5 13 1 4 1 \\
 12 9 4 15 10 16 3 4 14 6 3 9 11 12 13 10 6 8 1 15 1 7 16 14 11 5 8 2 13 7 2 5 \\
 12 9 3 15 10 16 3 1 14 1 6 9 11 12 13 10 4 6 7 15 8 4 16 14 11 5 7 2 13 8 2 5 \\
 12 9 4 15 10 16 3 4 14 7 3 9 11 12 13 10 1 7 1 15 6 8 16 14 11 5 2 6 13 2 8 5 \\
 12 9 4 15 10 16 3 4 14 6 3 9 11 12 13 10 6 2 7 15 2 8 16 14 11 5 7 1 13 1 8 5 \\
 12 9 4 15 10 16 2 4 14 2 6 9 11 12 13 10 3 6 7 15 3 8 16 14 11 5 7 1 13 1 8 5 \\
 12 8 4 15 10 16 9 4 14 7 8 6 11 12 13 10 9 7 6 15 5 3 16 14 11 3 5 2 13 1 2 1 \\
 12 8 6 15 10 16 9 3 14 6 8 3 11 12 13 10 9 2 7 15 2 5 16 14 11 4 7 5 13 1 4 1 \\
 12 7 4 15 10 16 9 4 14 7 6 8 11 12 13 10 9 6 2 15 8 2 16 14 11 5 3 1 13 1 3 5 \\
 12 5 7 15 10 16 9 5 14 1 7 1 11 12 13 10 9 3 6 15 8 3 16 14 11 6 4 2 13 8 2 4 \\[5pt]
 10 8 9 4 12 1 16 1 4 14 8 10 9 15 11 6 13 12 2 5 7 2 6 16 14 5 11 3 7 15 13 3 \\
 10 8 9 4 12 1 16 1 4 14 8 10 9 15 11 5 13 12 6 2 7 5 2 16 14 6 11 3 7 15 13 3 \\
 10 6 9 1 12 1 16 7 6 14 8 10 9 15 11 7 13 12 3 8 4 5 3 16 14 4 11 5 2 15 13 2 \\
 10 6 9 1 12 1 16 7 6 14 8 10 9 15 11 7 13 12 2 8 4 2 5 16 14 4 11 3 5 15 13 3 \\
 10 3 9 4 12 3 16 7 4 14 8 10 9 15 11 7 13 12 6 8 1 5 1 16 14 6 11 5 2 15 13 2 \\
 10 3 9 4 12 3 16 2 4 14 2 10 9 15 11 6 13 12 8 1 7 1 6 16 14 5 11 8 7 15 13 5 \\
 10 6 9 3 12 4 16 3 6 14 4 10 9 15 11 5 13 12 1 7 1 5 8 16 14 2 11 7 2 15 13 8 \\
 10 8 5 3 12 9 16 3 5 14 8 10 7 15 11 9 13 12 6 1 7 1 4 16 14 6 11 4 2 15 13 2 \\
 10 6 7 3 12 9 16 3 6 14 7 10 8 15 11 9 13 12 1 5 1 8 4 16 14 5 11 4 2 15 13 2 \\
 10 5 7 4 12 9 16 5 4 14 7 10 6 15 11 9 13 12 8 6 1 3 1 16 14 3 11 8 2 15 13 2 \\[5pt]
 9 7 5 1 6 1 10 16 5 7 9 6 13 14 15 11 12 10 2 3 4 2 8 3 16 4 13 11 14 12 15 8 \\
 9 3 5 7 4 3 10 16 5 4 9 7 13 14 15 11 12 10 6 1 2 1 8 2 16 6 13 11 14 12 15 8 \\
 9 7 3 1 6 1 3 16 10 7 9 6 13 14 15 11 12 5 2 10 4 2 8 5 16 4 13 11 14 12 15 8 \\
 9 3 6 7 1 3 1 16 10 6 9 7 13 14 15 11 12 5 2 10 4 2 8 5 16 4 13 11 14 12 15 8 \\
 9 6 2 7 1 2 1 16 6 10 9 7 13 14 15 11 12 3 4 5 10 3 8 4 16 5 13 11 14 12 15 8 \\
 9 4 2 7 3 2 4 16 3 10 9 7 13 14 15 11 12 5 6 1 10 1 8 5 16 6 13 11 14 12 15 8 \\
 9 4 5 1 6 1 4 16 5 10 9 6 13 14 15 11 12 7 2 3 10 2 8 3 16 7 13 11 14 12 15 8 \\
 9 3 5 2 6 3 2 16 5 10 9 6 13 14 15 11 12 7 4 1 10 1 8 4 16 7 13 11 14 12 15 8 \\
 9 2 3 4 2 10 3 16 4 7 9 11 12 14 15 13 10 7 1 6 1 5 8 11 16 12 6 5 14 13 15 8 \\
 9 3 1 4 1 3 10 16 4 6 9 11 12 14 15 13 6 10 2 7 5 2 8 11 16 12 5 7 14 13 15 8 \\

\normalsize
\noindent
$p(3,9)=3$. The following is a complete list of all its solutions:
\medskip \\
 1 9 1 6 1 8 2 5 7 2 6 9 2 5 8 4 7 6 3 5 4 9 3 8 7 4 3 \\
 1 9 1 2 1 8 2 4 6 2 7 9 4 5 8 6 3 4 7 5 3 9 6 8 3 5 7 \\
 3 4 7 9 3 6 4 8 3 5 7 4 6 9 2 5 8 2 7 6 2 5 1 9 1 8 1 \\

\noindent
$p(3,9)=5$. The following is a complete list of all its solutions:
\medskip \\
 1 10 1 6 1 7 9 3 5 8 6 3 10 7 5 3 9 6 8 4 5 7 2 10 4 2 9 8 2 4 \\
 1 10 1 2 1 4 2 9 7 2 4 8 10 5 6 4 7 9 3 5 8 6 3 10 7 5 3 9 6 8 \\
 4 10 1 7 1 4 1 8 9 3 4 7 10 3 5 6 8 3 9 7 5 2 6 10 2 8 5 2 9 6 \\
 8 1 10 1 3 1 9 6 3 8 4 7 3 10 6 4 9 5 8 7 4 6 2 5 10 2 9 7 2 5 \\
 1 3 1 10 1 3 4 9 6 3 8 4 5 7 10 6 4 9 5 8 2 7 6 2 5 10 2 9 8 7 \\

\DeclareFixedFont{\smallrm}{T1}{ptm}{m}{n}{9pt}
\noindent
$p(3,17)=13440$. The following is a partial list of all its solutions:
\medskip \\ \smallrm
 17 15 3 16 9 10 3 1 12 1 3 1 13 14 9 6 10 15 17 5 16 12 6 11 9 5 13 10
 14 6 7 5 8 15 12 11 17 16 7 4 13 8 2 14 4 2 7 11 2 4 8 \\
 17 9 11 16 2 3 15 2 10 3 2 9 14 3 11 12 7 13 17 10 16 9 15 5 7 8 11 14
 12 5 10 13 7 6 8 5 17 16 15 4 6 12 14 8 4 13 1 6 1 4 1 \\
 17 11 14 16 4 10 1 15 1 4 1 9 6 11 4 13 10 14 17 6 16 9 12 15 8 11 6 10
 5 13 7 9 14 8 5 12 17 16 7 15 5 3 8 13 2 3 7 2 12 3 2  \\
 17 14 9 16 2 8 10 2 15 11 2 7 9 12 8 13 14 10 17 7 16 11 9 8 15 4 12 7
 10 13 4 14 5 11 6 4 17 16 5 12 15 6 3 13 5 1 3 1 6 1 3  \\
 17 4 13 16 2 8 4 2 15 3 2 4 14 3 8 6 13 3 17 11 16 12 6 8 15 9 5 14 10
 6 13 11 5 7 12 9 17 16 5 10 15 7 14 11 1 9 1 12 1 7 10 \\
 17 13 3 16 4 5 3 11 15 4 3 5 9 14 4 13 6 5 17 11 16 12 9 6 15 7 8 10 14
 13 6 11 9 7 12 8 17 16 10 2 15 7 2 14 8 2 1 12 1 10 1 \\
 17 7 9 16 2 8 4 2 15 7 2 4 9 13 8 14 4 7 17 11 16 12 9 8 15 3 5 13 10 3
 14 11 5 3 12 6 17 16 5 10 15 13 6 11 1 14 1 12 1 6 10 \\
 17 4 9 16 2 3 4 2 15 3 2 4 9 3 13 14 10 5 17 11 16 12 9 5 15 6 8 10 13 5
 14 11 6 7 12 8 17 16 10 6 15 7 13 11 8 14 1 12 1 7 1 \\
 17 14 4 16 8 2 3 4 2 15 3 2 4 8 3 5 14 12 17 13 16 5 8 10 7 15 11 5 6 9 12
 14 7 13 10 6 17 16 11 9 7 15 6 12 1 10 1 13 1 9 11  \\
 17 4 14 16 5 6 4 7 12 15 5 4 6 8 13 7 5 14 17 6 16 12 8 7 10 15 11 2 13 9
 2 8 14 2 12 10 17 16 11 9 3 15 13 1 3 1 10 1 3 9 11 \\[5pt]
 4 17 12 8 9 4 16 1 13 1 4 1 8 15 9 12 14 7 11 17 5 8 13 16 9 7 5 10 12 15
 11 14 5 7 6 3 13 17 10 3 16 6 11 3 2 15 14 2 6 10 2 \\
 6 17 8 5 3 12 16 6 3 5 11 8 3 15 6 5 9 14 12 17 8 13 11 16 2 7 9 2 10 15 2
 12 14 7 11 13 9 17 4 10 16 7 1 4 1 15 1 14 4 13 10 \\
 9 17 5 13 10 1 16 1 5 1 9 4 12 15 5 10 4 13 14 17 9 4 11 16 2 12 10 2 7 15
 2 13 8 14 11 6 7 17 12 3 16 8 6 3 7 15 11 3 14 6 8 \\
 10 17 6 13 2 9 16 2 5 6 2 10 11 15 5 9 6 13 14 17 5 12 10 16 11 9 1 7 1 15
 1 13 8 14 12 7 11 17 3 4 16 8 3 7 4 15 3 12 14 4 8  \\
 4 17 11 8 12 4 16 1 6 1 4 1 8 15 11 6 10 12 13 17 14 8 6 16 9 2 11 10 2 15
 12 2 13 7 9 14 5 17 10 3 16 7 5 3 9 15 13 3 5 7 14  \\
 13 17 10 1 3 1 16 1 3 12 9 4 3 10 13 15 4 11 14 17 9 4 12 16 10 6 7 8 13 11
 9 15 6 14 7 12 8 17 5 6 16 11 7 2 5 8 2 15 14 2 5   \\
 11 17 8 10 5 3 16 9 12 3 5 8 11 3 10 15 5 9 14 17 8 12 13 16 11 10 2 9 6 2 7
 15 2 14 12 6 13 17 7 4 16 1 6 1 4 1 7 15 14 4 13  \\
 12 17 2 3 4 2 16 3 2 4 5 3 11 12 4 14 5 15 6 17 13 9 5 16 11 6 12 7 10 8 14
 9 6 15 13 7 11 17 8 10 16 9 1 7 1 14 1 8 13 15 10  \\
 13 17 8 11 3 1 16 1 3 1 10 8 3 12 13 11 14 15 5 17 8 10 6 16 5 9 12 11 13 6
 5 14 10 15 7 9 6 17 4 12 16 2 7 4 2 9 14 2 4 15 7  \\
 10 17 4 11 9 3 16 4 12 3 7 10 4 3 9 11 13 15 7 17 14 12 10 16 9 2 7 11 2 8 13
 2 6 15 12 14 5 17 8 6 16 1 5 1 13 1 6 8 5 15 14 \\[5pt]
 2 12 17 2 5 15 2 1 10 1 5 1 16 14 12 9 5 13 3 10 17 15 3 8 11 9 3 12 14 16 10
 13 8 6 7 9 11 15 17 4 6 8 7 14 4 13 16 6 11 4 7 \\
 8 2 17 10 2 15 1 2 1 8 1 5 16 13 10 14 9 5 8 11 17 15 12 5 3 10 9 13 3 16 14 11
 3 6 7 12 9 15 17 4 6 13 7 11 4 14 16 6 12 4 7  \\
 2 6 17 2 13 10 2 15 6 4 11 9 16 14 4 6 10 8 13 4 17 9 11 15 12 7 8 10 14 16 5 9
 13 7 11 8 5 12 17 15 3 7 5 14 3 1 16 1 3 1 12  \\
 3 4 17 5 3 13 4 15 3 5 10 4 16 14 1 5 1 6 1 13 17 10 12 15 6 11 7 8 14 16 9 6 10
 13 7 12 8 11 17 15 9 2 7 14 2 8 16 2 12 11 9  \\
 13 9 17 5 8 10 1 15 1 5 1 9 16 8 13 5 10 14 12 6 17 9 8 15 11 4 6 10 13 16 4 12
 14 6 7 4 11 2 17 15 2 3 7 2 12 3 16 14 11 3 7  \\
 6 1 17 1 12 1 9 6 15 8 10 3 16 14 6 3 9 12 8 3 17 10 13 11 15 5 9 8 14 16 12 5 10
 7 4 11 13 5 17 4 15 7 2 14 4 2 16 11 2 7 13  \\
 10 7 17 1 5 1 12 1 15 7 5 10 16 9 13 14 5 7 8 12 17 4 10 9 15 11 4 8 13 16 14 4 12
 9 6 3 8 11 17 3 15 6 13 3 2 14 16 2 6 11 2  \\
 12 1 17 1 2 1 10 2 15 6 2 9 16 12 3 14 6 10 3 13 17 9 3 6 15 11 12 7 10 16 14 9 8 13
 4 7 5 11 17 4 15 8 5 7 4 14 16 13 5 11 8  \\
 10 7 17 1 11 1 3 1 15 7 3 10 16 12 3 14 11 7 8 5 17 13 10 9 15 5 12 8 11 16 14 5 4 9
 6 13 8 4 17 12 15 6 4 9 2 14 16 2 6 13 2  \\
 2 12 17 2 5 9 2 11 15 1 5 1 16 1 12 9 5 14 4 11 17 13 10 4 15 9 7 12 4 16 8 11 14 10
 7 13 6 3 17 8 15 3 7 6 10 3 16 14 8 13 6 \\[5pt]
 13 7 1 17 1 14 1 8 4 7 16 10 11 4 13 15 8 7 4 12 14 17 10 2 11 8 2 16 13 2 9 15 12
 10 6 14 11 5 3 17 9 6 3 5 16 12 3 15 6 5 9 \\
 1 3 1 17 1 3 9 14 2 3 16 2 10 5 2 15 9 12 13 5 8 17 14 10 11 5 9 16 6 8 12 15 13 7
 10 6 11 14 8 17 4 7 6 12 16 4 13 15 11 7 4 \\
 8 4 5 17 11 6 4 14 5 8 16 4 6 10 5 15 11 7 8 6 13 17 14 12 10 7 2 16 11 2 9 15 2 7
 13 10 12 14 3 17 9 1 3 1 16 1 3 15 13 12 9 \\
 13 10 2 17 14 2 6 3 2 7 16 3 10 6 13 3 15 7 11 14 6 17 8 10 12 7 9 16 13 5 11 8 15
 4 14 5 9 12 4 17 8 5 11 4 16 1 9 1 15 1 12  \\
 11 4 5 17 14 9 4 10 5 13 16 4 11 8 5 9 15 12 10 14 7 17 8 13 11 9 6 16 7 10 12 8 15
 6 14 2 7 13 2 17 6 2 3 12 16 1 3 1 15 1 3  \\
 12 10 8 17 14 1 7 1 9 1 16 8 10 12 7 13 15 6 9 14 8 17 7 10 6 11 12 16 9 13 2 6 15
 2 14 5 2 11 3 17 4 5 3 13 16 4 3 5 15 11 4  \\
 12 3 4 17 14 3 8 4 9 3 16 2 4 12 2 8 15 2 9 14 10 17 13 5 8 11 12 16 9 5 7 10 15 6
 14 5 13 11 7 17 6 1 10 1 16 1 7 6 15 11 13  \\
 13 7 9 17 1 14 1 3 1 7 16 3 9 5 13 3 15 7 11 5 14 17 9 10 12 5 6 16 13 8 11 4 15 6
 10 14 4 12 8 17 6 4 11 2 16 10 2 8 15 2 12  \\
 13 5 9 17 1 14 1 5 1 7 16 3 9 5 13 3 15 7 11 3 14 17 9 10 12 7 6 16 13 8 11 4 15 6
 10 14 4 12 8 17 6 4 11 2 16 10 2 8 15 2 12  \\
 8 5 9 17 1 14 1 5 1 8 16 12 9 5 13 7 15 11 8 6 14 17 9 7 12 10 6 16 13 11 3 7 15 6
 3 14 10 12 3 17 4 11 13 2 16 4 2 10 15 2 4 \\[5pt]
 2 3 11 2 17 3 2 16 10 3 12 8 15 1 11 1 14 1 6 10 8 13 17 12 16 6 11 9 15 8 10 14
 6 5 7 13 12 9 4 5 17 16 7 4 15 5 14 9 4 13 7 \\
 2 13 8 2 17 10 2 16 12 4 7 8 15 9 4 13 10 14 7 4 8 12 17 9 16 11 7 10 15 13 5 6 14
 9 12 3 5 11 6 3 17 16 5 3 15 6 1 14 1 11 1 \\
 8 12 13 9 17 5 1 16 1 8 1 5 15 9 12 4 13 5 8 14 4 11 17 9 16 4 10 12 15 6 13 7 2
 11 14 2 6 10 2 7 17 16 3 6 15 11 3 7 10 14 3  \\
 8 12 13 9 17 5 1 16 1 8 1 5 15 9 12 4 13 5 8 14 4 11 17 9 16 4 10 12 15 2 13 7 2 11
 14 2 6 10 3 7 17 16 3 6 15 11 3 7 10 14 6  \\
 2 6 7 2 17 3 2 16 6 3 7 11 15 3 10 6 12 13 7 14 5 8 17 11 16 10 5 9 15 12 8 13 5 4
 14 11 10 9 4 8 17 16 12 4 15 13 1 9 1 14 1  \\
 2 13 6 2 17 5 2 16 12 6 7 5 15 9 10 13 6 5 7 11 14 12 17 9 16 10 7 4 15 13 8 11 4 9
 12 14 10 4 3 8 17 16 3 11 15 1 3 1 8 1 14  \\
 4 6 8 10 17 4 2 16 6 2 4 8 2 15 10 6 13 14 7 11 8 5 17 12 16 10 7 5 9 15 13 11 14 5
 7 1 12 1 9 1 17 16 3 11 13 15 3 14 9 12 3  \\
 2 6 13 2 17 5 2 16 6 4 10 5 7 15 4 6 13 5 14 4 7 10 17 12 16 8 11 9 7 15 13 3 10 14
 8 3 12 9 11 3 17 16 1 8 1 15 1 9 14 12 11  \\
 7 1 3 1 17 1 3 16 7 2 3 6 2 15 13 2 7 9 6 14 12 10 17 11 16 6 4 9 13 15 8 4 10 12 14
 11 4 9 5 8 17 16 13 10 5 15 12 11 8 14 5  \\
 2 6 11 2 17 5 2 16 6 8 10 5 12 15 11 6 13 5 8 14 7 10 17 9 16 12 11 8 7 15 13 3 10 9
 14 3 7 4 12 3 17 16 4 9 13 15 1 4 1 14 1 \\[5pt]
 16 11 7 3 9 17 12 3 13 10 7 3 15 11 9 14 4 16 7 12 10 4 13 17 9 11 4 5 15 8 14 10
 12 5 16 6 13 2 8 5 2 17 6 2 15 14 1 8 1 6 1 \\
 16 8 13 5 1 17 1 6 1 5 8 12 7 15 6 5 13 16 14 8 7 6 9 17 12 10 11 3 7 15 13 3 9 14
 16 3 10 12 11 4 2 17 9 2 4 15 2 10 14 4 11 \\
 16 12 6 3 10 17 1 3 1 6 1 3 11 15 12 10 6 16 14 8 4 13 9 17 11 4 10 12 8 15 4 7 9 14
 16 13 11 8 5 7 2 17 9 2 5 15 2 7 14 13 5  \\
 16 12 6 3 1 17 1 3 1 6 8 3 11 15 12 10 6 16 14 8 4 13 9 17 11 4 10 12 8 15 4 7 9 14
 16 13 11 10 5 7 2 17 9 2 5 15 2 7 14 13 5  \\
 16 6 4 2 11 17 2 4 6 2 9 7 4 15 10 6 11 16 14 7 9 13 8 17 12 10 5 7 11 15 9 8 5 14
 16 13 10 12 5 3 8 17 1 3 1 15 1 3 14 13 12  \\
 16 4 6 8 2 17 4 2 10 6 2 4 8 15 12 9 6 16 14 10 11 8 13 17 1 9 1 12 1 15 10 7 11 14
 16 9 13 5 3 7 12 17 3 5 11 15 3 7 14 5 13  \\
 16 9 4 2 10 17 2 4 13 2 8 9 4 6 15 10 12 16 14 8 6 9 13 17 7 11 10 6 8 12 15 3 7 14
 16 3 13 11 5 3 7 17 12 1 5 1 15 1 14 11 5  \\
 16 6 4 2 8 17 2 4 6 2 13 9 4 8 15 6 12 16 14 7 5 9 8 17 13 11 5 7 10 12 15 9 5 14 16
 7 3 11 13 10 3 17 12 1 3 1 15 1 14 11 10  \\
 16 2 3 10 2 17 3 2 11 9 3 4 13 14 10 15 4 16 8 9 11 4 7 17 12 10 13 8 14 9 7 15 11 5
 16 6 8 12 7 5 13 17 6 14 1 5 1 15 1 6 12  \\
 16 11 6 8 13 17 1 12 1 6 1 7 8 11 14 15 6 16 13 7 12 8 3 17 10 11 3 7 9 14 3 15 13 12
 16 10 2 5 9 2 4 17 2 5 14 4 10 15 9 5 4 \\[5pt]
 4 9 5 15 13 4 17 12 5 7 4 9 16 14 5 11 6 7 13 15 12 9 8 6 17 7 10 11 14 16 6 8 13
 12 1 15 1 10 1 11 8 3 17 14 2 3 16 2 10 3 2 \\
 10 2 5 15 2 4 17 2 5 6 4 10 16 14 5 4 6 9 12 15 13 8 10 6 17 7 11 9 14 16 8 12 3
 7 13 15 3 9 11 8 3 7 17 14 12 1 16 1 13 1 11 \\
 13 2 7 15 2 3 17 2 11 3 7 9 16 3 13 14 8 4 7 15 11 9 4 12 17 8 10 4 13 16 14 9 11
 6 8 15 12 10 5 1 6 1 17 1 5 14 16 6 10 12 5  \\
 3 12 7 15 3 1 17 1 3 1 7 4 16 9 12 14 4 10 7 15 13 4 8 9 17 11 5 12 10 16 14 8 5 9
 13 15 6 11 5 10 8 2 17 6 2 14 16 2 13 11 6  \\
 3 11 7 15 3 10 17 1 3 1 7 1 16 11 2 14 10 2 7 15 2 12 13 8 17 11 6 10 9 16 14 5 8 6
 12 15 13 5 9 4 6 8 17 5 4 14 16 12 9 4 13  \\
 1 5 1 15 1 9 17 5 2 11 7 2 16 5 2 9 13 14 7 15 6 11 8 12 17 9 7 6 10 16 13 8 14 11 6
 15 12 3 4 10 8 3 17 4 13 3 16 14 4 12 10  \\
 13 11 8 15 5 1 17 1 6 1 5 8 16 11 13 6 5 9 14 15 8 10 6 12 17 11 3 9 13 16 3 7 10 14
 3 15 12 9 2 7 4 2 17 10 2 4 16 7 14 12 4  \\
 8 1 13 1 15 1 17 7 2 8 10 2 16 14 2 7 13 9 8 5 15 10 12 7 17 5 11 9 14 16 13 5 10 6 4
 12 15 9 11 4 6 3 17 14 4 3 16 6 12 3 11  \\
 7 8 11 13 15 1 17 1 7 1 8 6 16 9 11 14 7 13 6 8 15 12 4 9 17 6 11 4 10 16 14 13 4 9
 12 2 15 5 2 10 3 2 17 5 3 14 16 12 3 5 10  \\
 4 12 6 3 15 4 17 3 9 6 4 3 16 5 12 7 6 14 9 5 15 13 10 7 17 5 11 12 9 16 8 7 14 10 2
 13 15 2 11 8 2 1 17 1 10 1 16 14 8 13 11 \\[5pt]
 4 12 6 13 7 4 14 17 10 6 4 11 7 15 12 16 6 13 3 10 7 14 3 11 8 17 3 12 9 15 10 13 
 16 8 5 11 14 1 9 1 5 1 8 17 2 15 5 2 9 16 2 \\
 4 10 11 13 6 4 14 17 7 12 4 6 10 15 11 16 7 13 6 8 9 14 12 10 7 17 11 5 8 15 9 13 
 16 5 3 12 14 8 3 5 9 2 3 17 2 15 1 2 1 16 1  \\
 13 1 3 1 6 1 3 17 14 11 3 6 4 15 13 16 9 4 6 7 10 11 4 14 12 17 9 7 13 15 8 10 16 
 11 5 7 9 12 14 8 5 2 10 17 2 15 5 2 8 16 12 \\
 10 1 3 1 12 1 3 17 7 4 3 10 13 15 4 16 7 12 14 4 6 9 10 11 7 17 13 6 8 15 12 9 16 
 14 6 11 2 8 5 2 13 9 2 17 5 15 8 11 14 16 5  \\
 11 9 1 14 1 3 1 17 13 3 8 9 11 3 15 16 12 2 14 8 2 9 13 2 11 17 10 6 8 12 15 7 16 
 14 6 4 13 10 5 7 4 6 12 17 5 4 15 7 10 16 5  \\
 12 9 1 14 1 6 1 17 7 13 11 9 6 12 15 16 7 10 14 6 8 9 11 13 7 17 12 3 10 8 15 3 16 
 14 11 3 5 13 8 10 4 2 5 17 2 4 15 2 5 16 4  \\
 10 1 7 1 14 1 6 17 8 13 7 10 11 6 15 16 12 8 7 14 6 5 10 13 11 17 8 5 9 12 15 3 16 
 5 14 3 11 13 9 3 4 2 12 17 2 4 15 2 9 16 4  \\
 12 9 1 6 1 14 1 17 13 8 6 9 10 12 15 16 11 6 8 4 14 9 13 10 4 17 12 8 11 4 15 7 16 
 3 10 14 13 3 5 7 11 3 2 17 5 2 15 7 2 16 5  \\
 12 8 1 13 1 2 1 17 2 14 8 2 3 12 15 16 3 13 9 8 3 7 10 11 14 17 12 6 9 7 15 13 16 
 10 6 11 5 7 9 14 4 6 5 17 10 4 15 11 5 16 4  \\
 9 3 1 4 1 3 1 17 4 3 9 14 10 4 15 16 11 6 8 13 9 7 12 10 6 17 14 8 11 7 15 6 16 13 
 10 12 8 7 5 2 11 14 2 17 5 2 15 13 12 16 5  \\

\normalsize
\DeclareFixedFont{\smallrm}{T1}{ptm}{m}{n}{8.25pt}
\noindent
$p(3,18)=54947$. The following is a partial list of all its solutions:
\medskip \\ \smallrm
 18 16 5 17 11 4 2 9 5 2 4 14 2 15 5 4 11 9 16 18 12 17 13 6 7 8 14 9 11 15 6 10 7 12 
 8 16 13 6 18 17 7 14 10 8 3 15 12 1 3 1 13 1 3 10 \\
 18 6 9 17 4 5 16 12 6 4 15 5 9 14 4 6 8 5 13 18 12 17 9 16 11 8 15 3 14 7 10 3 13 12 
 8 3 11 7 18 17 16 10 15 14 2 7 13 2 11 1 2 1 10 1  \\
 18 7 4 17 11 2 16 4 2 7 15 2 4 8 10 14 11 7 13 18 6 17 8 16 12 10 15 6 11 5 14 8 13 
 9 6 5 10 12 18 17 16 5 15 9 3 14 13 1 3 1 12 1 3 9  \\
 18 6 4 17 12 3 16 4 6 3 5 15 4 3 9 6 5 12 13 18 14 17 5 16 9 2 10 15 2 11 12 2 13 8 
 9 14 7 10 18 17 16 11 8 15 7 1 13 1 10 1 14 8 7 11  \\
 18 5 6 17 12 14 16 5 2 6 7 2 15 5 2 13 6 12 7 18 14 17 11 16 8 10 7 4 15 13 12 9 4 8 
 11 14 10 4 18 17 16 9 8 13 15 3 11 10 1 3 1 9 1 3  \\
 18 12 10 17 5 1 16 1 6 1 5 8 15 10 12 6 5 13 14 18 8 17 6 16 10 11 3 12 15 8 3 13 9 
 14 3 7 4 11 18 17 16 4 9 7 15 13 4 2 14 11 2 7 9 2  \\
 18 12 6 17 3 5 16 14 3 6 4 5 3 15 12 4 6 5 13 18 4 17 14 16 11 2 8 12 2 15 10 2 13 9 
 7 8 11 14 18 17 16 10 7 9 8 15 13 1 11 1 7 1 10 9  \\
 18 7 11 17 6 4 16 12 13 7 4 6 14 15 11 4 8 7 6 18 12 17 13 16 10 8 11 14 1 15 1 9 1 
 12 8 10 13 5 18 17 16 9 14 5 3 15 10 2 3 5 2 9 3 2 \\
 18 5 8 17 4 10 16 5 1 4 1 8 1 5 4 15 10 12 14 18 8 17 13 16 3 9 7 10 3 11 12 15 3 14 
 7 9 13 6 18 17 16 11 7 12 6 9 2 15 14 2 13 6 2 11  \\
 18 12 2 17 13 2 4 16 2 3 15 4 7 3 12 14 4 3 13 18 7 17 11 5 16 10 15 12 7 5 14 8 13 9 
 11 5 10 6 18 17 8 16 15 9 6 14 11 10 1 8 1 6 1 9  \\[5pt]
 12 18 2 11 9 2 16 13 2 7 8 10 17 12 9 11 14 7 15 8 18 13 10 16 9 7 12 11 8 6 17 14 5 10
 15 13 6 3 5 18 16 3 4 6 5 3 14 4 17 1 15 1 4 1 \\
 3 18 11 12 3 1 16 1 3 1 8 10 17 13 11 9 12 6 15 8 18 14 10 16 6 9 11 13 8 12 17 6 5 10
 15 9 14 7 5 18 16 13 4 2 5 7 2 4 17 2 15 14 4 7 \\
 1 18 1 6 1 2 12 16 2 11 6 2 17 4 13 15 8 6 4 12 18 11 14 4 16 8 9 7 13 10 17 15 12 11 8
 7 9 14 5 18 10 16 13 7 5 3 9 15 17 3 5 10 14 3  \\
 3 18 12 2 3 8 2 16 3 2 5 7 17 11 8 12 5 14 15 7 18 13 5 8 16 11 4 7 12 10 17 4 14 9 15
 13 4 11 6 18 10 16 1 9 1 6 1 14 17 13 15 10 6 9  \\
 13 18 3 5 6 10 3 16 8 5 3 6 17 11 13 5 10 8 6 15 18 12 9 14 16 11 8 10 13 1 17 1 9 1 12
 15 7 11 14 18 4 16 9 2 7 4 2 12 17 2 4 15 7 14  \\
 10 18 3 12 13 1 3 1 16 1 3 10 17 9 4 15 12 8 13 4 18 14 10 9 4 16 8 6 11 12 17 15 13 9 6
 8 14 7 5 18 11 6 16 2 5 7 2 15 17 2 5 14 11 7  \\
 13 18 3 9 10 1 3 1 16 1 3 4 17 9 13 10 4 15 8 12 18 4 14 9 6 16 10 8 13 11 17 6 12 15 5
 7 8 14 6 18 5 11 16 7 2 12 5 2 17 15 2 7 14 11  \\
 8 18 1 9 1 10 1 3 16 8 13 3 17 9 6 3 10 15 8 12 18 6 14 9 13 16 4 10 6 11 17 4 12 15 5
 7 4 14 13 18 5 11 16 7 2 12 5 2 17 15 2 7 14 11  \\
 12 18 3 7 11 1 3 1 16 1 3 7 17 12 8 13 11 9 15 7 18 14 10 8 6 16 12 9 11 13 17 6 8 10
 15 4 14 9 6 18 4 5 16 13 10 4 2 5 17 2 15 14 2 5  \\
 12 18 3 1 4 1 3 1 16 4 3 11 17 12 4 6 7 13 15 8 18 14 6 11 7 16 12 10 8 6 17 13 7 9 15
 11 14 8 10 18 5 2 16 9 2 13 5 2 17 10 15 14 5 9 \\[5pt]
 7 13 18 10 12 17 15 2 7 8 2 16 11 2 10 13 7 12 8 14 5 18 15 17 11 10 5 8 16 13 12 3 5
 9 14 3 11 6 15 3 18 17 4 9 6 16 1 4 1 14 1 6 4 9 \\
 6 7 18 12 13 17 15 6 1 7 1 16 1 10 6 11 12 7 13 14 4 18 15 17 10 4 9 11 16 12 4 5 13 8
 14 10 9 5 15 11 18 17 8 5 3 16 9 2 3 14 2 8 3 2 \\
 7 5 18 2 10 17 2 5 7 2 6 16 12 5 15 10 7 6 14 9 11 18 13 17 6 12 10 3 16 9 15 3 11 14
 8 3 13 4 12 9 18 17 4 8 11 16 15 4 14 1 13 1 8 1 \\
 5 6 18 13 2 17 5 2 6 9 2 16 5 11 15 6 7 13 8 9 14 18 10 17 7 11 12 8 16 9 15 13 7 10
 3 14 8 11 3 12 18 17 3 4 10 16 15 1 4 1 14 1 12 4 \\
 7 11 18 6 4 17 15 3 7 4 6 3 16 11 4 3 7 6 8 13 14 18 15 17 9 11 12 8 1 16 1 10 1 13 9
 14 8 5 15 12 18 17 10 5 9 2 16 13 2 5 14 2 12 10  \\
 10 5 18 13 4 17 12 5 7 4 15 10 16 5 4 14 7 13 9 12 6 18 10 17 7 11 15 6 9 16 14 13 12
 2 6 8 2 11 9 2 18 17 15 3 8 14 16 3 1 11 1 3 1 8  \\
 7 5 18 2 6 17 2 5 7 2 15 6 16 5 10 14 7 4 6 13 11 18 4 17 12 10 15 4 9 16 14 3 11 13 8
 3 10 12 9 3 18 17 15 8 11 14 16 13 9 1 12 1 8 1  \\
 7 2 18 13 2 17 8 2 7 12 9 15 16 5 11 8 7 13 14 5 9 18 12 17 8 5 11 15 10 16 9 13 1 14 1
 12 1 6 11 10 18 17 4 15 6 3 16 4 14 3 10 6 4 3  \\
 7 2 18 6 2 17 11 2 7 4 6 15 16 12 4 8 7 6 11 4 13 18 14 17 8 9 12 15 3 16 11 10 3 8 13
 9 3 14 5 12 18 17 10 15 5 9 16 1 13 1 5 1 14 10  \\
 5 6 18 10 11 17 5 1 6 1 15 1 5 16 10 6 11 4 13 14 8 18 4 17 12 10 15 4 11 8 16 3 13 9
 14 3 7 12 8 3 18 17 15 9 7 2 13 16 2 14 12 2 7 9  \\[5pt]
 8 13 3 18 14 15 3 12 17 8 3 5 1 16 1 13 1 5 8 14 12 15 18 5 11 2 17 10 2 13 16 2 9 12 14
 7 11 15 10 6 4 18 9 7 17 4 6 16 11 10 4 7 9 6 \\
 11 7 3 18 6 14 3 15 17 7 3 6 11 5 16 13 10 7 6 5 14 12 18 15 11 5 17 10 8 13 2 16 9 2 12
 14 2 8 10 15 4 18 9 13 17 4 8 12 16 1 4 1 9 1 \\
 6 13 7 18 15 5 12 6 17 14 7 5 8 11 6 13 16 5 7 12 15 8 18 10 14 11 17 3 9 13 8 3 12 16 10
 3 15 11 9 14 2 18 4 2 17 10 2 4 9 1 16 1 4 1  \\
 6 13 8 18 14 15 5 6 17 3 12 8 5 3 6 13 16 3 5 14 8 15 18 12 7 4 17 11 9 13 4 10 7 16 14 4
 12 15 9 11 7 18 10 1 17 1 2 1 9 2 16 11 2 10  \\
 13 11 3 18 6 14 3 15 17 2 3 6 2 11 13 2 10 16 6 12 14 7 18 15 5 11 17 10 13 7 5 8 12 9 16
 14 5 7 10 15 8 18 4 9 17 12 1 4 1 8 1 16 4 9  \\
 11 8 12 18 15 5 9 6 14 17 8 5 11 16 6 12 9 5 13 8 15 6 18 14 11 3 9 17 12 3 16 10 13 3 2 7
 15 2 14 4 2 18 10 7 4 17 13 16 1 4 1 7 1 10  \\
 5 8 12 18 13 14 5 15 2 17 8 2 5 16 2 12 4 6 13 8 14 4 18 15 6 11 4 17 12 10 16 6 13 9 3 14
 7 11 3 15 10 18 3 9 7 17 1 16 1 11 1 10 7 9  \\
 9 2 12 18 2 6 14 2 15 17 9 5 6 16 11 12 8 5 13 6 9 14 18 5 15 8 11 17 12 3 16 10 13 3 8 7
 14 3 11 4 15 18 10 7 4 17 13 16 1 4 1 7 1 10  \\
 11 9 13 18 15 4 10 5 14 17 4 9 11 5 16 4 13 10 6 5 15 9 18 14 11 6 12 17 10 3 13 16 6 3 7
 8 15 3 14 12 2 18 7 2 8 17 2 1 16 1 7 1 12 8  \\
 9 13 10 18 11 15 12 4 14 17 9 6 4 10 16 13 11 4 6 12 9 15 18 14 10 6 8 17 11 13 2 16 12 2
 7 8 2 15 14 3 5 18 7 3 8 17 5 3 16 1 7 1 5 1  \\[5pt]
 13 3 9 14 18 3 6 17 10 3 16 4 9 6 13 15 4 8 14 10 6 4 9 18 12 17 8 16 13 11 10 15 7
 14 1 8 1 12 1 5 7 11 18 17 16 5 2 15 7 2 12 5 2 11 \\
 10 3 5 14 18 3 6 17 5 3 16 10 11 6 5 1 15 1 14 1 6 13 10 18 11 17 12 16 2 8 9 2 15 14
 2 13 11 7 8 12 9 4 18 17 16 7 4 8 15 13 9 4 12 7 \\
 2 13 11 2 18 14 2 17 12 3 16 10 5 3 11 13 15 3 5 8 14 12 10 18 5 17 11 16 8 13 9 7 15
 10 12 14 4 8 6 7 9 4 18 17 16 6 4 7 15 1 9 1 6 1 \\
 2 7 11 2 18 14 2 17 4 7 16 12 8 4 11 6 15 7 4 13 14 8 6 18 12 17 11 16 10 6 8 9 15 13
 1 14 1 12 1 10 5 9 18 17 16 3 5 13 15 3 10 9 5 3 \\
 2 13 14 2 18 7 2 17 3 4 16 10 3 7 4 13 3 14 15 4 9 7 10 18 11 17 12 16 8 13 9 6 14 10
 15 5 11 8 6 12 9 5 18 17 16 6 8 5 11 1 15 1 12 1 \\
 14 1 5 1 18 1 13 17 5 3 16 4 9 3 5 14 4 3 10 15 13 4 9 18 11 17 12 16 8 10 14 6 9 7 13
 15 11 8 6 12 10 7 18 17 16 6 8 2 11 7 2 15 12 2 \\
 14 9 11 2 18 6 2 17 4 2 16 9 6 4 11 14 7 13 4 6 15 9 12 18 7 17 11 16 10 8 14 13 7 3 5
 12 15 3 8 10 5 3 18 17 16 13 5 8 12 1 10 1 15 1 \\
 11 3 5 14 18 3 13 17 5 3 10 16 11 2 5 15 2 8 14 2 13 10 7 18 11 17 8 12 16 9 7 15 10 14
 13 8 4 6 7 9 12 4 18 17 6 16 4 15 1 9 1 6 1 12 \\
 11 1 7 1 18 1 14 17 5 3 7 16 11 3 5 6 15 3 7 13 5 14 6 18 11 17 12 10 16 6 9 2 15 13 2
 8 14 2 10 12 9 4 18 17 8 16 4 13 15 10 9 4 12 8 \\
 14 8 11 4 18 10 13 17 4 12 8 16 5 4 11 14 10 15 5 8 13 7 12 18 5 17 11 10 16 7 14 9 6
 15 13 12 1 7 1 6 1 9 18 17 3 16 6 2 3 15 2 9 3 2 \\[5pt]
 17 2 12 6 2 18 16 2 3 11 6 13 3 15 7 12 3 6 17 14 10 11 7 16 18 13 5 9 12 15 7 10 5 11
 14 8 17 9 5 13 16 4 10 18 8 15 4 9 1 14 1 4 1 8 \\
 17 2 3 11 2 18 3 2 16 4 3 10 15 12 4 11 6 14 17 4 8 13 10 6 18 16 12 11 15 8 6 9 14 10
 5 13 17 7 8 12 5 9 16 18 15 7 5 14 1 13 1 9 1 7  \\
 17 7 12 8 2 18 4 2 16 7 2 4 8 15 11 12 4 7 17 14 9 8 10 13 18 16 11 3 12 15 9 3 6 10 14
 3 17 13 11 6 9 5 16 18 10 15 6 5 1 14 1 13 1 5  \\
 17 2 12 6 2 18 4 2 16 11 6 4 10 13 15 12 4 6 17 9 14 11 8 10 18 16 5 13 12 9 15 8 5 11
 10 14 17 7 5 9 8 13 16 18 3 7 15 1 3 1 14 1 3 7  \\
 17 2 9 6 2 18 4 2 16 11 6 4 9 13 15 10 4 6 17 12 14 11 9 7 18 16 10 13 5 8 15 7 12 11
 5 14 17 10 8 7 5 13 16 18 3 12 15 8 3 1 14 1 3 1 \\
 17 2 13 6 2 18 7 2 16 1 6 1 10 1 7 15 13 6 17 8 14 11 7 10 18 16 5 12 8 9 13 15 5 11
 10 14 17 8 5 9 12 4 16 18 3 11 4 15 3 9 14 4 3 12 \\
 17 7 12 10 1 18 1 11 1 7 16 4 9 15 10 12 4 7 17 11 14 4 9 13 18 10 6 16 12 15 8 11 9
 6 3 14 17 13 3 8 6 5 3 18 16 15 2 5 8 2 14 13 2 5 \\
 17 5 12 2 7 18 2 5 8 2 16 4 7 5 15 12 4 8 17 14 7 4 9 11 18 13 8 16 12 10 15 6 9 3 14
 11 17 3 6 13 10 3 9 18 16 6 15 11 1 14 1 10 1 13 \\
 8 17 11 3 5 18 16 3 13 8 5 3 15 6 11 14 5 7 8 17 6 12 13 16 18 7 11 6 15 9 14 10 2 7
 12 2 13 17 2 9 16 4 10 18 15 14 4 12 1 9 1 4 1 10 \\
 7 17 8 13 4 18 16 2 7 4 2 8 15 2 4 11 7 13 14 17 8 12 6 16 18 9 5 11 15 6 10 13 5 14
 12 9 6 17 5 11 16 10 3 18 15 9 3 12 14 1 3 1 10 1  \\[5pt]
 12 8 2 16 4 2 18 7 2 4 8 17 14 12 4 7 15 6 9 8 16 11 13 7 6 18 12 14 9 17 10 6 15
 11 5 3 13 16 9 3 5 10 14 3 18 11 5 17 15 1 13 1 10 1 \\
 1 3 1 16 1 3 18 6 8 3 11 17 2 14 6 2 15 8 2 7 16 6 11 12 13 18 8 7 14 17 9 10 15 5
 11 7 12 16 13 5 9 4 10 14 18 5 4 17 15 12 9 4 13 10  \\
 11 12 7 16 10 3 18 13 14 3 7 17 11 3 12 10 2 15 7 2 16 13 2 14 11 18 10 12 1 17 1
 9 1 15 8 13 5 16 14 6 4 9 5 8 18 4 6 17 5 15 4 9 8 6  \\
 11 2 6 16 2 4 18 2 14 6 4 17 11 13 3 4 6 15 3 12 16 8 3 14 11 18 9 13 5 17 8 10 12
 15 5 7 9 16 14 8 5 13 10 7 18 12 9 17 1 15 1 7 1 10 \\
 4 7 13 16 10 4 18 5 14 7 4 17 9 5 6 10 13 7 15 5 16 6 9 14 11 18 10 12 6 17 13 3 9
 8 15 3 11 16 14 3 12 1 8 1 18 1 2 17 11 2 15 8 2 12 \\
 7 10 11 16 2 13 18 2 7 12 2 17 10 14 11 5 7 6 15 13 16 5 12 10 6 18 11 5 14 17 8 6
 9 13 15 12 1 16 1 8 1 4 9 14 18 3 4 17 8 3 15 4 9 3 \\
 13 1 6 1 16 1 18 14 4 6 9 17 11 4 13 3 6 15 4 3 9 16 14 3 11 18 12 7 13 17 9 10 8
 15 5 7 11 14 16 12 5 8 10 7 18 2 5 17 2 15 8 2 12 10 \\
 9 5 2 13 16 2 18 5 2 6 9 17 14 5 7 8 6 13 15 12 9 16 7 6 8 18 10 14 11 17 7 13 12 8
 15 3 4 10 16 3 11 4 14 3 18 12 4 17 10 1 15 1 11 1 \\
 9 1 6 1 16 1 18 11 8 6 9 17 14 10 7 12 6 8 15 11 9 16 7 13 10 18 8 14 12 17 7 11 3 5
 15 10 3 13 16 5 3 12 14 4 18 5 2 17 4 2 15 13 2 4 \\
 9 3 12 13 8 3 18 16 14 3 9 17 1 8 1 12 1 13 15 6 9 10 8 14 16 18 6 11 12 17 4 13 10
 6 15 4 5 7 14 11 4 16 5 10 18 7 2 17 5 2 15 11 2 7 \\[5pt]
 1 13 1 8 1 2 12 18 2 15 9 2 8 14 3 13 16 17 3 12 9 8 3 7 11 15 18 10 14 13 9 7 12 
 16 5 17 11 6 10 7 5 15 4 14 6 18 5 4 11 10 16 6 4 17 \\
 9 13 1 4 1 6 1 18 4 15 9 12 6 4 8 13 16 17 11 6 9 14 5 8 12 15 18 10 5 13 11 7 8 
 16 5 17 14 12 10 7 2 15 11 2 3 18 2 7 3 10 16 14 3 17  \\
 1 11 1 2 1 12 2 18 6 2 13 15 7 11 8 6 16 17 12 10 7 14 6 8 13 11 18 15 7 9 10 12 8 
 16 5 17 14 4 13 9 5 10 4 15 3 18 5 4 3 9 16 14 3 17  \\
 1 6 1 2 1 12 2 18 6 2 13 15 7 5 8 6 16 17 12 5 7 14 11 8 13 5 18 15 7 9 10 12 8 16 
 11 17 14 4 13 9 3 10 4 15 3 18 11 4 3 9 16 14 10 17  \\
 10 3 1 11 1 3 1 18 5 3 12 10 15 7 5 11 16 17 6 14 5 7 10 12 13 6 18 11 15 7 8 9 6 
 16 14 17 12 4 13 8 2 9 4 2 15 18 2 4 8 14 16 9 13 17  \\
 9 3 1 4 1 3 1 18 4 3 9 12 15 4 8 13 16 17 11 2 9 14 2 8 12 2 18 10 15 13 11 7 8 16 
 6 17 14 12 10 7 5 6 11 13 15 18 5 7 6 10 16 14 5 17   \\
 5 1 7 1 8 1 5 18 14 9 7 12 5 8 13 15 16 17 7 9 4 10 8 14 12 4 18 11 13 9 4 15 10 16 
 6 17 3 12 14 11 3 6 13 10 3 18 2 15 6 2 16 11 2 17    \\
 8 9 1 4 1 6 1 18 4 8 12 9 6 4 14 15 16 17 8 6 13 9 2 12 7 2 18 11 2 14 10 15 7 16 13 
 17 12 5 3 11 7 10 3 5 14 18 3 15 13 5 16 11 10 17    \\
 4 12 10 2 3 4 2 18 3 2 4 6 3 10 12 15 16 17 6 9 13 14 7 8 10 6 18 12 11 9 7 15 8 16 
 13 17 14 5 7 9 11 8 1 5 1 18 1 15 13 5 16 14 11 17   \\
 4 12 10 7 13 4 2 18 14 2 4 7 2 10 12 6 16 17 13 7 15 9 6 14 10 11 18 12 5 6 8 9 13 16 
 5 17 15 11 14 8 5 9 1 3 1 18 1 3 8 11 16 3 15 17  \\

\normalsize
\DeclareFixedFont{\smallrm}{T1}{ptm}{m}{n}{7.75pt}
\noindent
$p(3,19)=249280$. The following is a partial list of all its solutions:
\medskip \\ \smallrm
 19 17 13 18 4 11 8 2 16 4 2 9 15 2 4 8 13 11 14 17 19 9 18 12 8 16 5 7 15 11 13 9 5
 14 10 7 12 17 5 6 19 18 16 7 15 10 6 3 14 12 1 3 1 6 1 3 10 \\
 19 17 10 18 5 2 11 9 2 16 5 2 12 10 14 15 5 9 11 17 19 13 18 4 10 12 16 9 4 14 11
 15 6 4 7 13 8 17 12 6 19 18 7 16 14 8 6 15 3 13 7 1 3 1 8 1 3 \\
 19 17 3 18 11 7 3 9 2 16 3 2 10 7 2 15 11 9 14 17 19 7 18 10 13 4 16 9 11 12 4 15 6
 14 10 4 8 17 13 6 19 18 12 16 5 8 6 15 14 1 5 1 13 1 8 12 5 \\
 19 17 5 18 6 1 8 1 5 1 16 6 10 7 5 8 14 15 6 17 19 7 18 10 8 13 11 16 3 7 12 14 3 15
 10 9 3 17 11 13 19 18 4 12 16 9 14 4 2 15 11 2 4 13 2 9 12 \\
 19 17 4 18 13 8 1 4 1 7 1 16 4 10 8 15 12 7 13 17 19 14 18 8 10 7 11 2 16 12 2 15 13
 2 9 10 14 17 11 6 19 18 12 5 9 16 6 15 3 5 11 14 3 6 9 5 3 \\
 19 17 3 18 10 12 3 6 14 1 3 1 16 1 6 10 15 7 12 17 19 6 18 14 5 7 10 11 13 16 5 12
 15 7 8 9 5 17 14 11 19 18 13 8 4 9 16 2 15 4 2 11 8 2 4 9 13 \\
 19 17 6 18 1 11 1 7 1 6 10 15 12 16 9 7 6 11 14 17 19 10 18 7 9 12 4 15 13 11 16 4
 10 14 9 8 4 17 12 5 19 18 13 15 8 5 3 16 14 2 3 5 2 8 3 2 13 \\
 19 17 2 18 6 2 3 7 2 11 3 6 14 16 3 7 15 13 6 17 19 11 18 7 4 12 9 14 10 4 16 13 15
 11 4 8 9 17 12 10 19 18 14 5 8 13 9 16 15 5 10 12 1 8 1 5 1 \\
 19 17 6 18 3 7 8 10 3 6 12 15 3 7 16 8 6 14 10 17 19 7 18 12 8 4 11 15 13 10 4 16 14
 5 9 4 12 17 11 5 19 18 13 15 9 5 2 14 16 2 11 1 2 1 9 1 13 \\
 19 16 14 18 5 6 17 7 2 11 5 2 6 15 2 7 5 14 16 6 19 11 18 7 17 13 4 10 8 15 12 4 14
 11 9 16 4 8 10 13 19 18 17 12 9 15 8 3 1 10 1 3 1 13 9 3 12 \\[5pt]
 4 19 13 5 18 4 7 11 17 5 4 12 16 6 7 5 13 10 15 11 6 19 7 18 12 14 17 6 10 16 13 11
 2 9 15 2 8 12 2 10 14 19 18 9 17 8 16 1 3 1 15 1 3 9 8 14 3 \\
 5 19 12 7 18 1 5 1 17 1 9 7 5 11 16 12 10 15 13 7 9 19 6 18 14 11 17 10 12 6 9 16 13
 15 8 4 6 11 10 14 4 19 18 8 17 4 13 3 16 15 2 3 8 2 14 3 2 \\
 2 19 4 2 18 11 2 4 17 1 9 1 4 1 10 16 5 11 15 13 9 19 5 18 14 10 17 12 5 11 9 6 16 13
 15 7 10 8 6 14 12 19 18 7 17 6 8 13 3 16 15 7 3 12 14 8 3 \\
 13 19 10 3 18 4 9 3 11 17 4 3 16 10 13 4 9 7 12 15 11 19 14 18 10 7 9 17 13 16 8 12
 11 7 1 15 1 14 1 8 6 19 18 5 12 17 16 6 8 5 2 15 14 2 6 5 2 \\
 5 19 2 3 18 2 5 3 2 17 4 3 5 10 16 4 8 15 13 9 4 19 14 18 10 8 11 17 7 9 12 16 13 15
 8 10 7 14 11 9 6 19 18 12 7 17 13 6 16 15 11 1 14 1 6 1 12 \\
 9 19 10 12 18 6 11 3 13 17 9 3 6 10 16 3 12 15 11 6 9 19 13 18 10 14 5 17 8 12 11 16
 5 15 4 7 13 8 5 4 14 19 18 7 4 17 8 2 16 15 2 7 1 2 1 14 1 \\
 12 19 3 13 18 6 3 10 2 17 3 2 6 12 2 9 16 13 10 6 15 19 7 18 14 9 12 17 5 10 7 13 11
 16 5 9 15 8 7 14 5 19 18 4 11 17 8 1 4 1 16 1 15 4 14 8 11 \\
 7 19 8 2 18 13 2 4 7 2 17 8 4 16 12 10 7 4 15 13 8 19 6 18 14 9 10 12 17 6 16 11 5 13
 15 9 6 10 5 14 12 19 18 11 5 9 17 16 3 1 15 1 3 1 14 11 3 \\
 8 19 5 3 18 6 9 3 5 8 17 3 6 16 5 13 9 7 8 6 15 19 12 18 14 7 9 10 17 13 16 2 11 7 2
 12 15 2 10 14 4 19 18 13 11 4 17 16 12 10 4 1 15 1 14 1 11 \\
 7 19 9 3 18 11 6 3 7 8 17 3 9 6 16 10 7 11 8 15 6 19 9 18 13 14 10 8 17 11 12 16 1 5 1
 15 1 10 13 5 14 19 18 12 4 5 17 2 16 4 2 15 13 2 4 14 12 \\[5pt]
 11 9 19 2 12 10 2 17 14 2 7 9 11 15 18 16 10 12 7 8 13 9 19 14 11 17 7 10 8 15 12 3 16
 18 13 3 4 8 14 3 6 4 19 17 5 15 4 6 13 16 5 1 18 1 6 1 5 \\
 12 7 19 11 5 2 14 17 2 7 5 2 15 12 18 11 5 7 13 4 16 14 19 10 4 17 12 11 15 4 9 3 13
 18 10 3 14 16 8 3 9 6 19 17 15 10 13 8 6 1 9 1 18 1 16 6 8 \\
 2 13 19 2 5 12 2 17 10 14 5 11 15 3 18 13 5 3 12 10 16 3 19 11 14 17 9 6 15 13 10 12
 7 18 6 11 9 16 8 14 7 6 19 17 15 4 9 8 7 1 4 1 18 1 16 4 8 \\
 2 8 19 2 12 10 2 17 9 14 8 11 13 7 18 15 10 12 9 8 16 7 19 11 14 17 13 10 9 7 12 15 1
 18 1 11 1 16 5 14 13 6 19 17 5 3 4 15 6 3 5 4 18 3 16 6 4 \\
 9 4 19 5 14 10 4 17 8 5 9 4 11 15 18 5 10 8 13 14 9 16 19 6 11 17 8 10 12 15 6 3 13 18
 14 3 11 6 16 3 7 12 19 17 2 15 13 2 7 1 2 1 18 1 12 16 7 \\
 8 10 19 1 5 1 14 1 17 8 5 15 10 11 18 16 5 6 8 9 13 14 19 10 6 11 17 15 12 9 7 6 16 18
 13 4 14 11 7 9 4 12 19 15 17 4 7 3 13 16 2 3 18 2 12 3 2 \\
 11 1 19 1 14 1 4 13 17 8 12 4 11 15 18 16 4 10 8 14 9 13 19 12 11 5 17 8 10 15 9 5 16
 18 14 13 12 5 7 10 9 6 19 3 17 15 7 3 6 16 2 3 18 2 7 6 2 \\
 2 5 19 2 3 14 2 5 3 17 13 6 3 5 18 15 9 16 6 1 14 1 19 1 13 6 9 17 10 11 12 15 7 18 16
 14 9 8 13 10 7 11 19 12 4 17 8 15 7 4 10 16 18 11 4 8 12 \\
 3 4 19 8 3 14 4 10 3 17 12 4 8 2 18 15 2 16 10 2 14 8 19 12 1 13 1 17 1 10 9 15 11 18
 16 14 12 5 7 13 9 6 19 5 11 17 7 15 6 5 9 16 18 13 7 6 11 \\
 9 11 19 1 13 1 4 1 14 17 9 4 6 11 18 15 4 16 13 6 9 7 19 14 10 11 6 17 12 7 8 15 13 18
 16 10 3 7 14 8 3 12 19 5 3 17 10 15 8 5 2 16 18 2 12 5 2 \\[5pt]
 18 3 11 19 12 3 14 10 6 3 16 13 7 17 11 6 15 12 10 18 7 14 6 19 8 13 11 16 7 10 12 17
 15 8 2 9 14 2 18 13 2 5 8 19 16 9 4 5 15 17 1 4 1 5 1 9 4 \\
 18 9 2 19 5 2 13 14 2 7 5 9 16 10 17 15 5 7 11 18 13 9 14 19 10 7 6 8 12 16 11 15 17
 6 13 10 8 14 18 4 6 12 11 19 4 8 16 15 3 4 17 1 3 1 12 1 3 \\
 18 14 10 19 3 5 11 6 3 15 16 5 3 10 6 17 14 5 11 18 7 6 12 19 10 15 13 16 7 4 11 14 9
 17 4 12 7 8 18 4 13 15 9 19 16 1 8 1 12 1 2 17 9 2 13 8 2 \\
 18 2 6 19 2 5 14 2 9 6 16 5 13 11 15 17 6 5 9 18 10 14 12 19 4 11 13 16 9 4 15 10 8 17
 4 12 14 11 18 7 13 8 10 19 16 3 15 7 12 3 8 17 1 3 1 7 1 \\
 18 14 10 19 1 11 1 13 1 15 3 16 9 10 3 17 14 11 3 18 6 13 9 19 10 15 12 6 16 11 5 14 9
 17 6 13 5 8 18 12 7 15 5 19 4 16 8 2 7 4 2 17 12 2 4 8 7 \\
 18 12 1 19 1 7 1 14 3 15 9 16 3 7 12 17 3 8 13 18 9 7 14 19 11 15 8 12 16 4 9 10 13 17
 4 8 11 14 18 4 6 15 10 19 5 16 13 6 11 2 5 17 2 10 6 2 5 \\
 18 3 5 19 6 3 14 2 5 3 2 6 16 2 5 17 15 8 6 18 1 14 1 19 1 11 8 12 13 16 10 7 15 17 9 8
 14 11 18 7 12 10 13 19 9 4 16 7 15 11 4 17 10 12 9 4 13 \\
 18 12 14 19 7 1 4 1 13 1 15 4 7 16 12 17 4 14 2 18 7 2 13 19 2 5 15 12 11 9 16 5 14 17
 10 8 13 5 18 9 11 6 15 19 8 10 3 16 6 9 3 17 11 8 3 6 10 \\
 18 3 10 19 7 3 4 14 12 3 15 4 7 10 16 17 4 11 8 18 7 12 14 19 10 13 15 8 1 11 1 16 1 17
 12 9 8 14 18 13 6 11 15 19 5 9 2 6 16 2 5 17 2 13 6 9 5 \\
 18 5 1 19 1 14 1 5 9 2 16 10 2 5 15 2 17 11 9 18 14 4 10 19 12 13 4 16 9 11 15 4 7 10 17
 14 8 12 18 13 7 11 6 19 16 8 15 3 7 6 12 3 17 13 8 3 6 \\[5pt]
 2 6 9 2 19 11 2 1 6 1 10 1 9 12 16 6 18 11 15 17 4 10 9 14 19 4 12 8 13 11 4 16 10 7 15
 18 8 17 14 12 5 7 13 3 19 8 5 3 16 7 15 3 5 14 18 17 13 \\
 8 9 4 10 19 11 1 4 1 8 1 9 4 13 10 16 18 11 8 17 15 9 5 14 19 10 6 13 5 11 12 7 16 6 5
 18 15 17 14 7 6 13 3 12 19 2 3 7 2 16 3 2 15 14 18 17 12 \\
 12 3 4 6 19 3 7 4 10 3 6 13 4 12 7 16 18 6 15 10 17 14 7 9 19 13 12 1 11 1 10 1 16 9 15
 18 14 8 17 13 11 5 2 9 19 2 8 5 2 16 15 14 11 5 18 8 17 \\
 7 8 4 2 19 5 2 4 7 2 8 5 4 12 16 17 7 5 18 8 15 3 11 14 19 3 12 13 6 3 9 16 10 17 11 6
 15 18 14 12 9 13 6 10 19 1 11 1 16 1 9 17 15 14 10 13 18 \\
 9 10 6 4 19 5 12 13 4 6 9 5 10 4 16 17 6 5 18 12 9 13 15 10 19 14 1 8 1 11 1 16 12 17 7
 13 8 18 15 3 14 11 7 3 19 8 2 3 16 2 7 17 2 11 15 14 18 \\
 4 11 5 8 19 4 10 1 5 1 4 1 8 11 5 16 17 10 18 13 14 8 3 15 19 11 3 12 10 6 3 9 16 13 17
 14 6 18 7 15 12 9 2 6 19 2 7 13 2 16 14 9 17 12 7 15 18 \\
 12 8 5 1 19 1 2 1 5 2 8 10 2 12 5 16 17 13 18 8 4 14 10 15 19 4 12 7 9 11 4 13 16 10 17
 7 14 18 9 15 6 11 3 7 19 13 3 6 9 16 3 14 17 11 6 15 18 \\
 10 11 5 1 19 1 2 1 5 2 9 10 2 11 5 16 17 13 18 7 9 14 10 15 19 11 6 7 12 8 9 13 16 6 17
 7 14 18 8 15 6 12 3 4 19 13 3 8 4 16 3 14 17 4 12 15 18 \\
 5 12 2 3 19 2 5 3 2 7 10 3 5 8 12 13 17 7 18 16 14 10 8 15 19 7 6 12 9 13 11 8 10 6 17
 14 16 18 9 15 6 4 11 13 19 1 4 1 9 1 14 4 17 16 11 15 18 \\
 4 5 8 12 19 4 2 5 10 2 4 8 2 5 13 16 12 17 18 10 8 11 15 7 19 14 9 6 13 12 10 7 16 11 6
 17 9 18 15 7 14 6 13 3 19 11 9 3 1 16 1 3 1 17 15 14 18 \\[5pt]
 11 17 8 2 10 19 2 13 18 2 3 8 11 16 3 10 5 15 3 17 8 13 5 14 11 19 10 18 5 12 16 9 7 15
 4 13 6 17 14 4 7 9 12 6 4 19 18 16 7 15 6 9 1 14 1 12 1 \\
 8 17 12 9 7 19 4 13 18 8 11 4 7 9 16 12 4 15 8 17 7 13 11 9 14 19 5 18 12 6 10 16 5 15
 11 13 6 17 5 14 1 10 1 6 1 19 18 3 16 15 2 3 10 2 14 3 2 \\
 12 17 1 3 1 19 1 3 18 2 5 3 2 12 16 2 5 11 15 17 9 14 5 4 13 19 12 18 4 11 9 16 10 4 15
 8 14 17 13 7 9 11 6 10 8 19 18 7 16 6 15 14 13 8 10 7 6 \\
 6 17 7 2 8 19 2 6 18 2 7 4 13 8 6 16 4 11 7 17 15 4 8 14 9 19 13 18 10 11 12 1 16 1 9 1
 15 17 14 10 13 11 5 12 9 19 18 3 5 16 10 3 15 14 5 3 12 \\
 1 17 1 3 1 19 11 3 18 12 2 3 13 2 6 16 2 5 11 17 15 6 12 5 14 19 13 18 6 5 11 9 16 10 7
 12 15 17 8 14 13 9 7 4 10 19 18 8 4 16 7 9 15 4 14 10 8 \\
 6 17 1 8 1 19 1 6 18 3 7 10 8 3 6 16 13 3 7 17 15 8 10 11 14 19 7 18 5 12 13 9 16 10 5
 11 15 17 4 14 5 9 12 4 13 19 18 11 4 16 2 9 15 2 14 12 2 \\
 1 17 1 9 1 19 2 12 18 2 3 4 2 9 3 15 4 16 3 17 12 4 6 9 14 19 11 18 13 6 10 15 7 12 16 8
 6 17 11 14 7 10 13 5 8 19 18 15 7 5 11 16 10 8 14 5 13 \\
 5 17 13 9 10 19 5 1 18 1 3 1 5 9 3 10 13 16 3 17 15 6 8 9 14 19 10 18 6 12 13 8 11 7 16
 6 15 17 4 14 8 7 12 4 11 19 18 2 4 7 2 16 15 2 14 12 11 \\
 6 8 17 13 10 19 2 6 18 2 8 11 2 16 6 10 5 13 15 8 17 14 5 11 7 19 10 18 5 12 16 13 7 9 15
 11 14 1 17 1 7 1 12 9 4 19 18 16 3 4 15 14 3 9 4 12 3 \\
 10 13 17 2 5 19 2 4 18 2 5 10 4 16 6 13 5 4 15 11 17 6 10 12 14 19 8 18 6 13 16 11 9 3 15
 8 12 3 17 14 7 3 9 11 8 19 18 16 7 12 15 1 9 1 14 1 7 \\[5pt]
 17 2 13 6 2 7 19 2 5 14 6 10 18 7 5 16 13 6 17 15 5 7 10 1 14 1 19 1 11 12 13 18 16 10
 9 15 17 3 8 14 11 3 12 4 9 3 19 8 4 16 18 15 11 4 9 12 8 \\
 17 8 5 1 11 1 19 1 5 14 8 10 18 13 5 16 11 2 17 8 2 15 10 2 14 4 19 13 11 12 4 18 16 10
 9 4 17 15 6 14 7 13 12 3 9 6 19 3 7 16 18 3 6 15 9 12 7 \\
 17 8 2 9 12 2 19 1 2 1 8 1 18 9 14 16 3 12 17 8 3 5 15 9 3 13 19 5 11 14 12 18 16 5 10 7
 17 4 15 13 11 6 4 7 14 10 19 4 6 16 18 7 11 13 15 6 10  \\
 17 9 5 1 12 1 19 1 5 6 11 9 18 14 5 16 6 12 17 7 8 9 11 6 15 13 19 7 14 8 12 18 16 10 11
 7 17 3 8 13 15 3 4 14 10 3 19 4 2 16 18 2 4 13 2 10 15  \\
 1 17 1 12 1 2 19 6 2 14 9 2 18 4 6 16 12 7 4 17 9 6 15 4 14 7 19 13 11 12 9 18 16 7 10 8
 5 17 15 14 11 13 5 3 8 10 19 3 5 16 18 3 11 8 15 13 10 \\
 1 17 1 12 1 2 19 6 2 8 4 2 18 14 6 4 12 16 8 17 4 6 15 11 5 13 19 8 14 12 5 18 9 10 16 11
 5 17 15 13 7 3 9 14 10 3 19 11 7 3 18 16 9 13 15 10 7  \\
 8 2 17 11 2 5 19 2 13 8 14 5 18 16 6 11 7 5 8 15 17 6 13 10 7 14 19 11 6 12 16 18 7 9 10
 15 13 1 17 1 14 1 12 9 4 10 19 16 3 4 18 15 3 9 4 12 3  \\
 8 6 17 1 11 1 19 1 6 8 14 10 18 16 9 6 11 5 8 12 17 15 10 5 9 14 19 13 11 5 16 18 12 10
 9 7 2 15 17 2 14 13 2 7 4 12 19 16 3 4 18 7 3 15 4 13 3  \\
 10 2 17 6 2 13 19 2 9 5 6 10 18 16 14 5 8 6 9 13 17 5 10 15 11 8 19 12 9 14 16 18 4 13 8
 7 11 4 17 15 12 3 4 7 14 3 19 16 11 3 18 7 1 12 1 15 1  \\
 10 3 17 9 2 3 19 2 7 3 2 10 18 9 14 16 7 12 8 15 17 11 10 9 7 13 19 8 5 14 12 18 16 11 5
 15 8 4 17 13 5 6 4 12 14 11 19 4 6 16 18 15 1 13 1 6 1  \\[5pt]
 18 4 5 8 12 14 4 19 5 6 17 4 8 15 5 16 6 12 3 18 14 8 3 6 11 13 3 19 17 15 12 9 16 10 7
 14 11 2 18 13 2 9 7 2 10 15 17 19 11 16 7 9 1 13 1 10 1 \\
 18 7 2 10 8 2 14 19 2 7 17 11 13 8 10 16 3 7 15 18 3 14 8 11 3 10 13 19 17 12 6 9 16 5
 15 11 14 6 18 5 13 9 12 4 6 5 17 19 4 16 15 9 1 4 1 12 1  \\
 18 9 1 6 1 14 1 19 7 13 6 9 17 15 10 16 7 6 3 18 14 9 3 13 7 10 3 19 12 15 17 11 16 8 4
 14 10 13 18 4 5 12 8 11 4 15 5 19 17 16 2 8 5 2 12 11 2  \\
 18 4 11 13 5 9 4 19 12 14 5 4 17 15 11 9 5 13 16 18 7 12 10 8 14 9 11 19 7 15 17 13 8 10
 12 16 7 2 18 14 2 8 6 2 10 15 3 19 17 6 3 1 16 1 3 1 6  \\
 18 9 5 2 13 14 2 19 5 2 12 9 17 15 5 10 11 6 13 18 14 9 16 12 6 7 10 19 11 15 17 6 13 7
 8 14 12 10 18 16 11 7 3 8 4 15 3 19 17 4 3 1 8 1 4 1 16  \\
 18 11 1 6 1 14 1 19 2 13 6 2 17 11 2 8 4 6 15 18 14 4 16 13 8 11 4 19 7 12 17 9 10 8 15
 14 7 13 18 16 5 9 12 10 7 3 5 19 17 3 15 9 5 3 10 12 16  \\
 18 11 1 5 1 14 1 19 2 5 6 2 17 11 2 5 12 6 15 18 14 10 16 4 6 11 13 19 4 12 17 9 10 4 15
 14 7 8 18 16 13 9 12 10 7 3 8 19 17 3 15 9 7 3 13 8 16  \\
 18 9 2 12 10 2 5 19 2 14 11 9 5 17 15 10 12 7 5 18 16 9 11 8 14 7 10 19 13 12 15 17 8 7
 11 4 6 16 18 14 4 8 13 6 3 4 15 19 3 17 6 1 3 1 16 1 13  \\
 18 7 2 14 10 2 3 19 2 7 3 4 15 17 3 10 4 7 14 18 13 4 16 8 12 9 10 19 15 6 11 17 8 14 13
 9 6 12 18 16 5 8 11 6 15 9 5 19 13 17 12 1 5 1 11 1 16  \\
 18 10 3 5 13 14 3 19 8 5 3 4 10 17 15 5 4 8 13 18 14 4 16 10 12 7 8 19 11 6 15 17 13 7 9
 14 6 12 18 16 11 7 2 6 9 2 15 19 2 17 12 1 11 1 9 1 16  \\[5pt]
 11 12 1 3 1 4 1 3 19 16 4 3 11 9 12 4 17 18 13 5 14 8 15 9 11 5 16 12 19 10 8 5 13 9 17 
 14 18 7 15 8 10 2 6 16 2 7 13 2 19 6 14 10 17 7 15 18 6  \\
 5 3 11 9 6 3 5 14 19 3 16 6 5 9 11 15 17 18 6 12 4 13 14 9 8 4 11 16 19 10 4 15 12 8 17 
 13 18 14 7 1 10 1 8 1 16 12 7 15 19 13 2 10 17 2 7 18 2  \\
 12 9 3 5 13 14 3 15 19 5 3 9 16 12 4 5 17 18 13 4 14 9 8 15 4 7 12 10 19 16 11 8 13 7 17 
 14 18 6 10 15 8 7 11 1 6 1 16 1 19 10 2 6 17 2 11 18 2   \\
 6 1 7 1 10 1 12 6 19 14 7 9 16 13 6 10 17 18 7 12 4 9 15 11 14 4 10 13 19 16 4 9 12 8 17 
 11 18 2 15 14 2 13 8 2 5 3 16 11 19 3 5 8 17 3 15 18 5   \\
 9 10 1 7 1 12 1 13 19 15 9 7 10 16 11 6 17 18 12 7 9 13 6 10 14 15 11 4 19 6 16 12 4 8 17 
 13 18 4 11 14 2 15 8 2 5 3 2 16 19 3 5 8 17 3 14 18 5    \\
 8 10 1 3 1 4 1 3 19 8 4 3 10 16 14 4 17 18 8 7 5 11 15 10 12 13 5 7 19 14 16 9 5 11 17 7 
 18 12 15 13 2 9 6 2 14 11 2 16 19 6 12 9 17 13 15 18 6   \\
 10 12 1 8 1 3 1 11 19 3 15 10 8 3 12 16 17 18 5 11 14 8 10 9 5 13 15 12 19 7 5 11 16 9 17 
 14 18 7 2 13 6 2 15 9 2 7 4 6 19 16 14 4 17 13 6 18 4    \\
 8 1 10 1 7 1 3 14 19 8 3 15 7 10 3 16 17 18 8 13 7 9 14 2 10 12 2 15 19 2 11 9 16 13 17 4 
 18 14 12 6 4 9 11 15 5 4 6 13 19 16 5 12 17 6 11 18 5    \\
 8 1 3 1 9 1 3 11 19 8 3 15 12 7 9 16 17 18 8 11 14 7 10 6 9 12 13 15 19 7 6 11 16 10 17 14 
 18 6 12 2 13 5 2 15 10 2 4 5 19 16 14 4 17 5 13 18 4     \\
 10 1 3 1 13 1 3 7 19 14 3 10 8 15 12 7 17 18 13 5 16 8 10 7 14 5 11 12 19 15 8 5 13 9 17 6 
 18 16 11 14 12 2 6 9 2 15 4 2 19 6 11 4 17 9 16 18 4

\end{document}